\RequirePackage{etex}
\documentclass[11pt]{article}
\usepackage{geometry} 

\usepackage{amsmath}
\usepackage{amssymb}
\usepackage{amsthm}
\usepackage{tikz-cd}
\usepackage{xspace}
\usepackage{mathtools} %
\usepackage{braket} %
\usepackage[linesnumbered,vlined]{algorithm2e}
\usepackage{xcolor} %
\usepackage{footnote} 
\usepackage{environ} %
\usepackage{suffix} %
\usepackage{apptools}
\usepackage{hyperref}
\usepackage{float}
\usepackage[capitalise]{cleveref} %
\usepackage[shortlabels]{enumitem}
\usepackage{hyphenat}

\usepackage{thmtools}
\usepackage{thm-restate}

\usepackage{natbib}
\setcitestyle{numbers,square}

\makesavenoteenv{tabular}%
\makesavenoteenv{table}%

\DeclareMathOperator{\poly}{poly}

\Crefname{theorem}{Theorem}{Theorems}
\Crefname{lemma}{Lemma}{Lemmas}
\Crefname{figure}{Figure}{Figures}
\Crefname{claim}{Claim}{Claims}
\Crefname{observation}{Observation}{Observations}

\newtheorem{theorem}{Theorem}[section]
\newtheorem{lemma}{Lemma}[section]

\newtheorem{claim}{Claim}[section]

\newtheorem{corollary}{Corollary}[section]

\newtheorem{question}{Question}[section]

\newenvironment{proofof}[1]{\proof[Proof of #1]}{\endproof}

\newcommand{\ID}{\mathsf{ID}}

\DeclarePairedDelimiter{\ceil}{\lceil}{\rceil}
\newcommand{\ip}[1]{\left}

\newcommand{\congest}{\ensuremath{\mathsf{CONGEST}}\xspace}

\newcommand{\clique}{\textnormal{\textsf{Congested-Clique}}\xspace}

\newcommand\bigO[1]{\ensuremath{{O}(#1)}}
\WithSuffix\newcommand\bigO*[1]{\ensuremath{{O}\left(#1\right)}}
\newcommand\tildeBigO[1]{\ensuremath{{\tilde{{O}}}(#1)}}
\WithSuffix\newcommand\tildeBigO*[1]{\ensuremath{{\tilde{{O}}}\left(#1\right)}}
\newcommand\littleO[1]{\ensuremath{{o}(#1)}}
\WithSuffix\newcommand\littleO*[1]{\ensuremath{{o}\left(#1\right)}}
\newcommand\tildeLittleO[1]{\ensuremath{{\tilde{{o}}}(#1)}}
\WithSuffix\newcommand\tildeLittleO*[1]{\ensuremath{{\tilde{{o}}}\left(#1\right)}}
\newcommand\bigOmega[1]{\ensuremath{{\Omega}(#1)}}
\WithSuffix\newcommand\bigOmega*[1]{\ensuremath{{\Omega}\left(#1\right)}}
\newcommand\tildeBigOmega[1]{\ensuremath{{\tilde{{\Omega}}}(#1)}}
\WithSuffix\newcommand\tildeBigOmega*[1]{\ensuremath{{\tilde{{\Omega}}}\left(#1\right)}}
\newcommand\littleOmega[1]{\ensuremath{{\omega}(#1)}}
\WithSuffix\newcommand\littleOmega*[1]{\ensuremath{{\omega}\left(#1\right)}}
\newcommand\tildeLittleOmega[1]{\ensuremath{{\tilde{{\omega}}}(#1)}}
\WithSuffix\newcommand\tildeLittleOmega*[1]{\ensuremath{{\tilde{{\omega}}{\left(#1\right)}}}}
\newcommand\bigTheta[1]{\ensuremath{{\Theta}(#1)}}
\WithSuffix\newcommand\bigTheta*[1]{\ensuremath{{\Theta}\left(#1\right)}}
\newcommand\tildeTheta[1]{\ensuremath{{\tilde{{\Theta}}}(#1)}}
\WithSuffix\newcommand\tildeTheta*[1]{\ensuremath{{\tilde{{\Theta}}\left(#1\right)}}}

\makeatletter
\newcommand*{\whp}{%
    \@ifnextchar{.}%
        {w.h.p}%
        {w.h.p.\@\xspace}%
}
\makeatother

\RestyleAlgo{boxruled}
\LinesNumbered
\SetKwFunction{AssignHelpers}{Assign-Helpers}
\SetKwFunction{TransmitSkeleton}{Transmit-Skeleton}
\SetKwFunction{LearnNeighborhood}{Learn-Neighborhood}
\SetKwFunction{LearnHelpers}{Learn-Helpers}
\SetKwFunction{ReassignSkeletons}{Reassign-Skeletons}

\title{Improved All-Pairs Approximate Shortest \\ Paths in Congested Clique\thanks{A preliminary version~\cite{bui2024improved} of this paper was presented at the 43rd ACM Symposium on Principles of Distributed Computing (PODC 2024), June 17--21, 2024, Nantes, France.}}

\author{\hspace{2cm}  Hong Duc Bui\footnote{National University of Singapore. Email: buihd@u.nus.edu} \and Shashwat Chandra\footnote{National University of Singapore. Email: shashwatchandra@u.nus.edu} \and Yi-Jun Chang\footnote{National University of Singapore. Email: cyijun@nus.edu.sg} \hspace{2cm}  \and Michal Dory\footnote{University of Haifa. Email: mdory@ds.haifa.ac.il} \and Dean Leitersdorf\footnote{National University of Singapore. Email: dean.leitersdorf@gmail.com}}

\usepackage{csquotes}\MakeOuterQuote"
\AddToHook{env/tikzcd/begin}{\catcode `\" 12}
\usepackage{catchfilebetweentags}

\usepackage[hypcap=false]{caption}
\allowdisplaybreaks

\date{}

\begin{document}

\maketitle
\begin{abstract}
In this paper, we present a new randomized $O(1)$-approximation algorithm for the All-Pairs Shortest Paths (APSP) problem in weighted undirected graphs that runs in just $O(\log \log \log n)$ rounds in the \clique model.

Before our work, the fastest algorithms achieving an $O(1)$-approximation for APSP in \emph{weighted} undirected graphs required $\poly(\log n)$ rounds, as shown by Censor-Hillel, Dory, Korhonen, and Leitersdorf (PODC 2019 \& Distributed Computing 2021). In the \emph{unweighted} undirected setting, Dory and Parter (PODC 2020 \& Journal of the ACM 2022) obtained $O(1)$-approximation in $\poly(\log \log n)$ rounds.

By terminating our algorithm early, for any given parameter $t \geq 1$, we obtain an $O(t)$-round algorithm that guarantees an $O\left(\log^{1/2^t} n\right)$ approximation in weighted undirected graphs. This tradeoff between round complexity and approximation factor offers flexibility, allowing the algorithm to adapt to different requirements. In particular, for any constant $\varepsilon > 0$, an $O\left(\log^\varepsilon n\right)$-approximation can be obtained in $O(1)$ rounds. Previously, $O(1)$-round algorithms were only known for $O(\log n)$-approximation, as shown by Chechik and Zhang (PODC 2022).

A key ingredient in our algorithm is a lemma that, under certain conditions, allows us to improve an $a$-approximation for APSP to an $O(\sqrt{a})$-approximation in $O(1)$ rounds. To prove this lemma, we develop several new techniques, including an $O(1)$-round algorithm for computing the $k$-nearest nodes, as well as new types of hopsets and skeleton graphs based on the notion of $k$-nearest nodes.
\end{abstract}

\newpage 
{ \normalsize	
\tableofcontents
}
\thispagestyle{empty} %
\newpage

\setcounter{page}{1}
\section{Introduction}\label{sec:intro}

The \emph{All-Pairs Shortest Paths (APSP)} problem is one of the most fundamental and well-studied problems in graph algorithms. It is particularly important in distributed computing due to its close connection to network routing. 

In this paper, we study the APSP problem in the distributed \clique model~\cite{lotker2005minimum}. This model consists of a fully connected communication network of \( n \) nodes, where each node can send \( O(\log n) \)-bit messages to any other node in synchronous rounds. The input graph \( G = (V, E) \) is arbitrary and given in a distributed manner: Each node initially knows the weights of its incident edges in \( G \), and the goal is for every node to compute its distances to all other nodes. 

We emphasize that while \( G \) can be any graph on \( n \) nodes, the communication network itself is a clique. The primary objective in this model is to minimize the number of communication rounds required to solve the problem. The \clique model has attracted significant attention in recent years, partly due to its connections to modern parallel computation frameworks such as the Massively Parallel Computation (MPC) model~\cite{karloff2010MapReduce}.

\subsection{Prior Work}
The problem of computing APSP in the \clique model has been extensively studied in recent years~\cite{censor2019algebraic,le2016further,censor2019sparse,censorhillel2019fast,dory2022exponentially,dory2021constant,chechik2022constant}. The earliest algorithms are based on matrix multiplication and require a polynomial number of rounds~\cite{censor2019algebraic,le2016further,censor2019sparse}. For instance, APSP in directed weighted graphs can be solved in \(\tilde{O}(n^{1/3})\) rounds, while for unweighted undirected graphs, algorithms with round complexity \(O(n^{0.158})\) are known~\cite{censor2019algebraic}. Designing faster algorithms for \emph{exact} APSP appears challenging, as shown in~\cite{dor2000all,DBLP:conf/spaa/KorhonenS18}: Any approximation factor better than two would already imply a fast matrix multiplication algorithm in the \clique model. The same hardness applies to \emph{any} approximation in directed graphs, which motivates the study of the problem in undirected graphs.

Hence, a natural approach for obtaining faster algorithms is to allow approximation algorithms. A recent line of work \cite{censorhillel2019fast,dory2022exponentially,dory2021constant, chechik2022constant} led to faster approximation algorithms for the problem. First, \cite{censorhillel2019fast} exploits \emph{sparse} matrix multiplication algorithms to obtain constant approximation for APSP in \emph{poly-logarithmic} number of rounds. This approach gives $(3+\varepsilon)$-approximation for APSP in weighted undirected graphs and $(2+\varepsilon)$-approximation in unweighted undirected graphs. In both cases, the running time is $O(\log^2{n}/\varepsilon)$ rounds. A subsequent work gives faster $\poly(\log \log n)$-round algorithms for a $(2+\varepsilon)$-approximation of APSP in unweighted undirected graphs \cite{dory2022exponentially}. The approach in \cite{dory2022exponentially} combines sparse matrix multiplication with fast construction of emulators, which are sparse graphs that approximate the distances in the input graph. Finally, recent works give $O(1)$-round algorithms for $O(\log{n})$-approximation for APSP in weighted undirected graphs via a fast construction of multiplicative spanners \cite{dory2021constant, chechik2022constant}. Here the goal is to construct a sparse graph of $O(n)$ edges that approximates the distances in the graph and broadcast it to the whole network. Because of the known tradeoffs between the number of edges and approximation guarantee of spanners, this approach leads to an $\Omega(\log{n})$-approximation.

To conclude, by now there are very fast algorithms that take just $O(1)$ rounds to compute an $O(\log{n})$-approximation of APSP. On the other hand, if our goal is to optimize the approximation factor, we can get an $O(1)$-approximation in $\poly(\log{n})$ rounds in weighted undirected graphs, or in $\poly(\log{\log{n}})$ rounds in unweighted undirected graphs. A natural goal is to obtain both $O(1)$-approximation and $O(1)$ rounds simultaneously.

\begin{question}\label{question1}
	Can we obtain $O(1)$-approximation of APSP in $O(1)$ rounds?
\end{question}

For weighted undirected graphs, the gap is even larger since the fastest \( O(1) \)-approximation algorithms take \( \poly(\log n) \) rounds. This naturally leads to the following question.

\begin{question} \label{question_wAPSP}
	Can we obtain $O(1)$-approximation of APSP on weighted graphs in $o(\log{n})$ rounds?
\end{question}

\subsection{Our Contribution}

In this work, we make progress in answering the above questions. Throughout the paper, we say that an event occurs \emph{with high probability} (w.h.p.) if it happens with probability at least $1 - 1/\poly(n)$, where $n$ denotes the number of nodes in the input graph. All our randomized algorithms are \emph{Monte Carlo} algorithms: They always meet the stated round complexity guarantees but may produce incorrect outputs with small probability.

\begin{theorem}[name=APSP approximation,restate=main]\label{approximate-APSP}
	For any constant $\varepsilon > 0$, a $(7^4 + \varepsilon)$-approximation of APSP in weighted undirected graphs can be computed w.h.p.\ in $O(\log \log \log n)$ rounds in the \clique model.
\end{theorem}

Our result provides the first sub-logarithmic round algorithm for the problem in \emph{weighted} undirected graphs, thereby answering \Cref{question_wAPSP} affirmatively. For \emph{unweighted} undirected graphs, our algorithm exponentially improves the number of rounds required to obtain an $O(1)$-approximation for APSP over the previous state of the art, which required $\poly(\log \log n)$ rounds~\cite{dory2022exponentially}.

Our approach also leads to tradeoffs between the running time and approximation.
More concretely, our algorithm starts with an $O(\log{n})$-approximation for APSP and successively improves it.
If we terminate the algorithm early, then we obtain the following tradeoff between the number of rounds needed and the approximation factor.

\begin{theorem}[name=Round-approximation tradeoff, restate=truncatedalgorithm]\label{tradeoff-APSP}
For any $t \geq 1$, an $O\left(\log^{2^{-t}} n\right)$-approximation of APSP in weighted undirected graphs can be computed w.h.p.\ in $O(t)$ rounds in the \clique model.
\end{theorem}

In particular, for any constant $\varepsilon > 0$, an $O\left(\log^\varepsilon n\right)$-approximation can be achieved in $O(1)$ rounds, making progress on \Cref{question1}.  Prior to this work, the best approximation factor attainable by $O(1)$-round algorithms was $O(\log n)$, first established for unweighted graphs~\cite{dory2021constant}, and later extended to weighted graphs~\cite{chechik2022constant}.

\section{Preliminaries}\label{sec:prelim}

In this section, we present the computational model, along with the basic definitions and tools.

\paragraph{The congested clique model.}
In the \clique model, we have a \emph{fully-connected} communication network of $n$ nodes, where nodes communicate by sending $O(\log{n})$-bit messages to each other node in synchronous rounds. We assume that each node $v$ is initially equipped with an $O(\log n)$-bit distinct identifier $\ID(v)$. By renaming, we may assume that the set of all identifiers is $\{1,2,\ldots, n\}$.

Given an input graph $G$ on $n$ nodes,
initially each node of the \clique knows its own input, i.e., the edges incident to it in $G$ and their weights, and at the end of the algorithm, it should know its output. For example, when computing shortest paths, a node $v$ should know its distances from other nodes.

We write \( \clique[B] \) to denote the variant of the \clique model in which the message size is \( O(B) \)~\cite{drucker2014power}. The standard \clique model corresponds to the special case \( B = \log n \). %

\subsection{Graph Terminology}

Throughout the paper, for the input graph $G = (V, E)$, we write $n = |V|$. %
Unless otherwise stated, we assume that the graph $G$ under consideration is \emph{simple}, \emph{undirected}, and \emph{weighted}. We write $w_{uv}$ to denote the weight of the edge $(u,v)$.
The number of \emph{hops} in a path $P$ refers to the number of edges in $P$. The \emph{length} of $P$ is the sum of the weights of its edges. The \emph{distance} $d_G(u, v)$ from node $u$ to node $v$ in $G$ is defined as the minimum length of any path from $u$ to $v$ in $G$. The subscript $G$ may be omitted when the graph is clear from context. The \emph{weighted diameter} of $G$ is the maximum distance $d_G(u, v)$ over all $u$ and $v$.

As is standard, we assume that edge weights are polynomially bounded \emph{positive} integers. However, as discussed in \Cref{subsec:zero_weight}, all our results on APSP approximation also extend to the case of polynomially bounded \emph{nonnegative} integer weights.%

For any node \( v \in V \) and any integer \( k \geq 0 \), the set  of \emph{\( k \)-nearest} nodes, denoted by \( N_k(v) \), is defined as the set of the \( k\) nodes \( u \) with the smallest values of \( d(u, v) \), breaking ties by node $\ID$s.

The \emph{$h$-hop distance}  between two nodes $u$ and $v$ is the minimum length of any path from $u$ to $v$ with at most $h$ hops. The set of {\( k \)-nearest} nodes \emph{within $h$ hops of $v$}, denoted by \( N_k^h(v) \), is the set of the \( k\) nodes \( u \) with the smallest $h$-hop distances to $v$, again breaking ties by node $\ID$s. The notations $d_G(u,v)$, $N_k(v)$, and $N_k^h(v)$ naturally generalize to \emph{directed} graphs.

\paragraph{Spanners.} A \emph{$k$-spanner} of a graph $G = (V, E)$ is a subgraph $G' = (V, E')$ on the same node set $V$ such that, for every pair of nodes $u, v \in V$, the distance in $G'$ satisfies
\[
d_{G'}(u, v) \leq k \cdot d_G(u, v).
\]
The parameter $k$ is referred to as the \emph{stretch} of the spanner $G'$.

\paragraph{Hopsets.}
A hopset is a set of edges added to a graph that enables shortcutting paths to reduce the number of hops. In this paper, we focus particularly on shortcutting paths to the $k$-nearest nodes and introduce the notion of a \emph{$k$-nearest $\beta$-hopset}. 

A graph $H$ is called a {$k$-nearest $\beta$-hopset} for a graph $G = (V, E)$ if it satisfies the following conditions:
\begin{itemize}
    \item $H$ is defined on the same node set $V$.
    \item For every pair of nodes $u, v \in V$, the distance is preserved: 
    \[
    d_G(u, v) = d_{G \cup H}(u, v).
    \]
    \item For all $u \in V$ and for all $v \in N_k(u)$, there exists a path in $G \cup H$ from $u$ to $v$ with at most $\beta$ hops and length exactly $d_G(u, v)$.
\end{itemize}

\paragraph{APSP.} In the All-Pairs Shortest Paths (APSP) problem, each node initially knows only its incident edges. The goal is for every node $u$ to learn the distance $d_G(u, v)$ to every other node $v$ by the end of the computation. 

In the approximate version, this requirement is relaxed so that each node learns an estimate $\delta(u, v)$ instead of the exact distance. An \emph{\(\alpha\)-approximation} satisfies
\[
d_G(u, v) \leq \delta(u, v) \leq \alpha \cdot d_G(u, v)
\]
for all pairs of nodes $u$ and $v$.

\paragraph{Matrix exponentiation.} The relationship between the APSP problem and the matrix multiplication problem in a tropical semiring is well-known --- specifically, exponentiation.

Let $R = (\mathbb {Z}_{\geq 0} \cup \{\infty \}, \oplus , \odot )$ be the tropical (or min-plus) semiring, with elements being the nonnegative integers, the addition operation $\oplus $ in the ring $R$ is defined by $x \oplus y = \min(x, y)$, and the multiplication $\odot $ is the usual addition of the integers by $x \odot y = x+y$.

For two matrices $A$ and $B$ with entries in $R$, define the distance product $A \star B$ as the matrix product over the tropical semiring, that is,
\[ (A \star B)[i, j] = \min _k (A[i, k]+B[k, j]). \]

Let \( A \) be the weighted adjacency matrix of a graph \( G=(V,E) \), with \( A[v, v] = 0 \) for all \( v \in V \). Then \( A^h \) represents the $h$-hop distances between all pairs of nodes. In particular, if \( h \) is at least the maximum number of hops in any shortest path in \( G \), then \( A^h \) equals the distance matrix---that is, \( A^h[u, v] \) equals the  distance from \( u \) to \( v \).

\subsection{Handling Zero Edge Weights}
\label{subsec:zero_weight}

Throughout the paper, we assume that all edge weights are positive integers.
However, the algorithm can be updated to handle the case where edge weights can be zero, using the following black-box reduction. 

\begin{restatable}{theorem}{thmZero} \label{thm:zero_weight_component_compression}
Suppose there exists an algorithm $\mathcal{A}$ that computes an $a$-approximation of APSP on weighted undirected graphs with \underline{positive} integer edge weights in $f(n)$ rounds in the \clique{} model. Then $\mathcal{A}$ can be extended to handle graphs with \underline{nonnegative} integer edge weights, yielding an $a$-approximation of APSP in $f(n) + O(1)$ rounds in the \clique{} model. Moreover, if $\mathcal{A}$ is deterministic, then the new algorithm remains deterministic.
\end{restatable}

The main idea is as follows. Given a graph $G$, we first identify clusters of nodes that are at distance $0$ from each other and compress each such cluster into a single node. We then run algorithm $\mathcal{A}$ on the resulting compressed graph and use its output to compute an $a$-approximation of APSP on the original graph $G$. See \Cref{sec:app_prelim} for the full proof of \Cref{thm:zero_weight_component_compression}.

\subsection{Useful Tools}

We use the following routing algorithms. First, we use Lenzen's routing algorithm \cite{Lenzen-routing}.
\begin{lemma}[name=\cite{Lenzen-routing}]\label{lemma:Lenzen-routing}
	There is a deterministic algorithm that delivers all the messages to the destinations within $O(1)$ rounds, where each node has $O(n)$ messages, each message consists of $O(\log n)$ bits of content, and each node is the target of $O(n)$ messages.
\end{lemma}

The following lemma from \cite{CensorHillel2021FastDA} generalizes Lenzen's routing algorithm to relax the condition that each node can only send $O(n)$ messages of $O(\log n)$ bits, as long as each node only receives $O(n)$ messages, and some technical conditions are satisfied. %

\begin{lemma}[{\cite[Corollary 7]{CensorHillel2021FastDA}}] \label{lemma:improved-routing}
	In the deterministic \clique model, given each node starts with $O(n \log n)$ bits of input, and at most $O(1)$ rounds have passed since the start of the algorithm, any routing instance where each node is the target of $O(n)$ messages of $O(\log n)$ bits can be performed in $O(1)$ rounds. %
\end{lemma}

Intuitively, if the messages sent by each node exhibit a high degree of duplication or redundancy, then the overall communication can be carried out efficiently. In the setting of \Cref{lemma:improved-routing}, each node begins with \( O(n \log n) \) bits of input, and at most \( O(1) \) rounds have passed since the start of the algorithm. This means that the messages sent by a node can be fully determined from just \( O(n \log n) \) bits of information. If a node wishes to send more than \( O(n) \) messages of \( O(\log n) \) bits each, we can allocate some helper nodes to assist in transmitting these messages.

As an example of how the lemma is applied, consider the case where a node wants to broadcast an \( O(n \log n) \)-bit message to all other nodes. Although the total communication volume is \( O(n^2 \log n) \) bits, the entire message is determined by just \( O(n \log n) \) bits of input. This redundancy allows us to apply \Cref{lemma:improved-routing} and complete the routing in \( O(1) \) rounds.

\section{Technical Overview} \label{sec:tech_overview}

Our algorithm relies crucially on the following lemma, which transforms an $a$-approximation into an $O(\sqrt{a})$-approximation in a constant number of rounds.  
Starting from an $O(\log n)$-approximation, which can be computed in $O(1)$ rounds using the result from~\cite{chechik2022constant}, we repeatedly apply the algorithm for $O(\log \log \log n)$ iterations.  
This yields successively better approximations: $O(\sqrt{\log n})$, $O(\log^{1/4} n)$, $O(\log^{1/8} n)$, and so on.  
However, the following lemma only applies when $a \in (\log d)^{\Omega(1)}$, where $d$ is the weighted diameter. Beyond this regime, a different strategy is required.

\begin{lemma}[name=Approximation factor reduction,restate=approximatereduction]\label{approximate-reduction}
	Let $a \geq 1$ be a given value such that $a  \in  O(\log n)$.
	Assume the weighted diameter $d$ of the  weighted undirected graph $G$ satisfies $\log d  \in  a^{O(1)}$.
There is a constant-round algorithm in the \clique model that, given an $a$-approximation of APSP, w.h.p.~computes a $15\sqrt{a}$-approximation of APSP. 
\end{lemma}

\subsection{Building Blocks} To prove the above lemma for approximation factor reduction, we use several building blocks, each of which may be of independent interest.

\paragraph{Fast computation of $\sqrt{n}$-nearest $\beta$-hopsets (\Cref{sec:hopset}).} As we will see, a main building block in our algorithm is a fast algorithm that allows each node to compute the distances to its $k$-nearest nodes. In order to get an efficient algorithm for this problem, we use hopsets.

A \emph{$\beta$-hopset} is a set of edges $H$ added to the graph $G$ such that, in $G \cup H$, for any pair of nodes $(u,v)$, there is a path of at most $\beta$ edges with length $d_G(u,v)$, where $d_G(u,v)$ is the distance between $u$ and $v$ in the original graph $G$. Hopsets allow us to consider only paths with a small number of hops when computing distances. For this reason, hopsets have found many applications for computing or approximating distances, especially in distributed and parallel settings. However, existing algorithms for constructing hopsets in the \clique model require at least a polylogarithmic number of rounds~\cite{censorhillel2019fast,DBLP:conf/opodis/Nazari19,dory2022exponentially}. 

As we aim for a faster running time, we need a different approach: We introduce a new type of hopset called a \emph{$k$-nearest $\beta$-hopset}. In this hopset, we are guaranteed to have $\beta$-hop paths that preserve the distances only for pairs of nodes $u$ and $v$ such that $v$ is among the $k$-nearest nodes to $u$, or vice versa. As our first goal is just to compute distances to the $k$-nearest nodes, a $k$-nearest $\beta$-hopset is good enough for our needs. 

We show a surprisingly simple $O(1)$-round algorithm for constructing a $k$-nearest hopset $H$ given an $a$-approximation for APSP. We focus on the case of $k=\sqrt{n}$. 

\begin{lemma}[name=Computing $\sqrt{n}$-nearest $\beta$-hopsets, restate=ahop]
	\label{approximate-to-hop}
Let $G$ be a weighted directed graph with weighted diameter $d$, and assume we are given an 
$a$-approximation of APSP. We can deterministically compute a $\sqrt{n}$-nearest $\beta$-hopset $H$ in $O(1)$ rounds, where $\beta \in O(a \log d)$. 
\end{lemma}

At a high level, to construct the hopset $H$, each node $u$ starts by computing a set $\tilde N_{\sqrt n}(u)$ of the approximate  $\sqrt{n}$-nearest nodes to $u$. This set is computed by taking the $k$-nearest nodes according to the given $a$-approximation. After that, $u$ learns from each node $v \in \tilde N_{\sqrt n}(u)$ about the $\sqrt{n}$ minimum-weight edges incident to $v$. The node $u$ computes shortest paths in the subgraph induced by the received edges, including also all the edges incident to $u$. Based on the computed shortest paths, $u$ adds edges to the hopset $H$. The algorithm takes $O(1)$ rounds, as each node just needs to learn $O(n)$ edges, which can done in $O(1)$ rounds using \Cref{lemma:improved-routing}. 

We prove that the computed edge set $H$ is a $\sqrt{n}$-nearest $\beta$-hopset for $\beta \in O(a \log d)$. To get an intuition for the proof, denote by $\ell(u)$ the smallest distance such that there are at least $\sqrt{n}$ nodes at distance at most $\ell(u)$ from $u$. We can show that the set $\tilde N_{\sqrt n}(u)$ of the approximate $\sqrt{n}$-nearest nodes already contains all nodes at distance at most $(\ell(u)-1)/a$ from $u$. We can use this fact to show that for any node $v$ in $u$'s $\sqrt{n}$-nearest nodes, we can use two edges of $G \cup H$ to cover at least $\frac{1}{a}$ fraction of the distance to $v$. We next show a triangle-inequality-like property for the values $\ell(u)$ that allows us to prove that we can use two other edges to cover $\frac{1}{a}$ fraction of the remaining distance, and so on. Therefore, within $O(a \log{d})$ hops, we reach $v$. %

\paragraph{Fast computation of the $k$-nearest nodes (\Cref{sec:knearest}).} Our next goal is to use the hopset to compute the $k$-nearest nodes for each node. Naively this may take $O(a \log{d})$ rounds, as this is the number of hops in the hopset. A faster algorithm can be obtained using fast matrix multiplication as shown in \cite{censorhillel2019fast,dory2022exponentially}. At a high level, this approach works in $i$ iterations where in iteration $i$ nodes learn about paths with $2^i$ hops. Using this approach in our case will take $O(\log{(a \log{d})})\subseteq O(\log{\log{n}})$ rounds, as our initial approximation is $a \in O(\log{n})$ and the weighted diameter is $d \in n^{O(1)}$. Recall that we assume that the edge weights are polynomial, as standard in this model. This round complexity is still too high. To get a faster algorithm, we identify cases in which the $k$-nearest nodes can be computed in just $O(1)$ rounds.

\begin{lemma}[name=Computing $k$-nearest nodes,restate=matrixexp]\label{fast-matrix-exp}
Let $k$, $h$, and $i$ be positive integers satisfying $k \in O(n^{1/h})$.
	Given a weighted directed graph $G$ and a %
 $k$-nearest $h^{i}$-hopset $H$ for $G$,
	 each node $v$ can deterministically compute the exact distances
	to all nodes in $N_k(v)$ in $O(i)$ rounds. 
\end{lemma}

We emphasize that this lemma allows for computing distances to the $k$-nearest nodes in $O(1)$ rounds even for \emph{non-constant} $k$. As an example, we can combine it with our $\sqrt{n}$-nearest $\beta$-hopset to compute distances to $k\in 2^{\Theta\left(\sqrt{\log{n}}\right)}$-nearest nodes in $O(1)$ rounds. To see this, we may choose $h \in \Theta\left(\sqrt{\log{n}}\right)$ and $i=4$, as the number of hops in our hopset is $\beta \in O(a \log n) \subseteq O(\log^2{n})$.

The main idea behind the proof of \Cref{fast-matrix-exp} is to distribute the edges of the graph $G \cup H$ among the nodes so that each relevant $h$-hop path from a node $u$ to any of its $k$-nearest nodes $v$ is known to some node $w$, which can then send the information back to $u$. The approach is inspired by algorithms for path listing~\cite{CensorHillel2021FastDA} and neighborhood collection~\cite{chechik2022constant}, designed for \emph{sparse graphs}---specifically, graphs without short cycles in one case~\cite{CensorHillel2021FastDA}, and low-degree graphs in the other~\cite{chechik2022constant}. 
This approach has been known for a long time and has been rediscovered several times. It plays a key role in the design of various $O(1)$-round algorithms in the $\clique$ model, for MST~\cite{jurdzinski2018mst}, coloring~\cite{chang2019complexity}, and, more recently, MIS in sparse graphs~\cite{censor2025mis}.

While in our case the input graph is not necessarily sparse, we can exploit the fact that we are only interested in a sparse part of the output---namely, the $k$-nearest nodes---to design an efficient algorithm. To do so, we distribute the edges in a certain way that indeed allows to collect the distances to the $k$nearest nodes efficiently. %

\paragraph{Skeleton graphs (\Cref{sec:skeleton}).} Finally, our goal is to extend the computed distances to the $k$-nearest nodes into approximate distances between all pairs of nodes. To achieve this, we construct a \emph{skeleton graph} $G_S$ with $\tilde{O}(n/k) \ll n$ nodes, such that an approximate solution to APSP on $G_S$ can be used to compute an approximate solution to APSP on the original graph $G$.

\begin{lemma}[name=Skeleton graphs -- simplified version, restate=skel] \label{lemma_skeleton_simplified}
Let $k$ be a positive integer. Suppose each node $u$ knows the distances to all nodes in $N_k(u)$ in a weighted undirected graph $G$. Then, w.h.p., we can construct a weighted undirected graph $G_S$ over a subset $V_S \subseteq V$ of size $O\left(\frac{n \log k}{k}\right)$ in $O(1)$ rounds. Moreover, given an $l$-approximation of APSP on $G_S$, we can deterministically compute a $7l$-approximation of APSP on $G$ in $O(1)$ rounds.
\end{lemma}

At a high level, \Cref{lemma_skeleton_simplified} reduces the task of computing APSP on $G$ to computing APSP on a smaller graph $G_S$. The smaller the skeleton graph $G_S$, the easier it becomes to approximate APSP on $G$. For instance, if $G_S$ has fewer than $\sqrt{n}$ nodes, then we can broadcast all of its edges in $O(1)$ rounds and compute exact APSP on $G_S$, which then yields an $O(1)$-approximation for APSP on $G$.

Other types of skeleton graphs have been used previously in the computation of shortest paths. The most standard example~\cite{ullman1990high} we are aware of is to sample a set of $\tilde{O}(n/k)$ nodes and connect pairs of sampled nodes that are at distance at most $\tilde{O}(k)$ from each other. This approach can lead to a round complexity that depends on $k$. Moreover, since $\tilde{O}(n/k)$ can still be large, it is unclear how to efficiently compute distances from the set of sampled nodes, even if the paths are short.

We take a different approach. Instead of connecting sampled nodes based on their distance in $G$, we connect pairs that are linked by a path of $O(1)$ edges with a specific structure in $G$, after augmenting $G$ by adding edges from each node to its $k$-nearest nodes, where the weight of each added edge corresponds to the known shortest-path distance. By leveraging the structure of these short paths and the fact that they consist of only $O(1)$ edges, we show that the skeleton graph $G_S$ can be constructed in $O(1)$ rounds, independent of $k$. Furthermore, we prove that using this skeleton graph only incurs a constant-factor loss in the approximation.

\subsection{APSP Approximation in \texorpdfstring{$O(\log \log n)$}{O(log log n)} Rounds}
The ingredients described above already suffice to obtain an $O(1)$-approximation for APSP in general graphs within $O(\log \log n)$ rounds, as follows:
\begin{enumerate}
\item Use the algorithm of \cite{chechik2022constant} to compute an $O(\log n)$-approximation of APSP in $O(1)$ rounds.
\item Apply \Cref{approximate-to-hop} with $a \in O(\log n)$ and $d \in n^{O(1)}$ to construct a $\sqrt{n}$-nearest $O(a \log d) \subseteq O(\log^2 n)$-hopset in $O(1)$ rounds.
\item Apply \Cref{fast-matrix-exp} with $k = \sqrt{n}$, $h = 2$, and $i \in O(\log \log n)$ to compute distances to the $\sqrt{n}$-nearest nodes in $O(\log \log n)$ rounds.
\item Apply \Cref{lemma_skeleton_simplified} with $k = \sqrt{n}$ to construct a skeleton graph $G_S = (V_S, E_S)$ with $O\left(\sqrt{n} \log n\right)$ nodes in $O(1)$ rounds.
\item Use the algorithm of \cite{chechik2022constant} to compute a 3-spanner of $G_S$ with $O(|V_S|^{1.5}) \subseteq O(n)$ edges and broadcast it to the entire network to get a 3-approximation of APSP on $G_S$ in $O(1)$ rounds.
\item Finally, by \Cref{lemma_skeleton_simplified}, this 3-approximation of APSP on $G_S$ can be translated to a 21-approximation of APSP on $G$ in $O(1)$ rounds.
\end{enumerate}
We compute a 3-spanner of $G_S$ instead of broadcasting its full topology because $G_S$ can have up to $O\left(|V_S|^2\right) \subseteq O(n \log^2 n)$ edges, which would require $O(\log^2 n)$ rounds to broadcast in the standard \clique model. However, if we are in the $\clique[\log^3 n]$ model, then the increased bandwidth allows us to broadcast the entire topology of $G_S$ in $O(1)$ rounds. This enables exact computation of APSP on $G_S$, which can then be translated into a 7-approximation of APSP on $G$.

 \subsection{APSP Approximation in \texorpdfstring{$O(\log \log \log n)$}{O(log log log n)} Rounds}

The bottleneck of the above algorithm lies in the step where we apply \Cref{fast-matrix-exp} to compute distances to the $k = \sqrt{n}$ nearest nodes. To reduce the round complexity to $O(\log \log \log n)$, we can no longer afford to compute distances to that many nodes. In particular, to perform this step in $O(1)$ rounds, we must restrict $k$ to $O(n^{1/h})$ for some $h$ such that $h^{O(1)}$ is the hop bound of the given hopset. We first show how to overcome this limitation and obtain an $O(\log \log \log n)$-round algorithm for graphs whose weighted diameter $d$ satisfies $\log d \in a^{O(1)}$, and then extend the result to general graphs.

\paragraph{APSP approximation in small weighted diameter graphs (\Cref{sec:algos}).}
To overcome this limitation, we proceed in iterations: In each iteration, we improve the APSP approximation and gradually increase the value of $k$ that we can handle. As a result, both the number of nodes in the skeleton graph $G_S$ and the number of hops in the hopset decrease over time.

More concretely, we combine the above ingredients to prove \Cref{approximate-reduction}, which shows that an $a$-approximation can be improved to an $O(\sqrt{a})$-approximation in $O(1)$ rounds, provided that the weighted diameter $d$ of $G$ satisfies $\log d \in a^{O(1)}$.

\begin{enumerate}
	\item Start with a given $a$-approximation of APSP. %
	\item Compute a $\sqrt{n}$-nearest $O(a\log{d})$-hopset using \Cref{approximate-to-hop}.
	\item Use the computed hopset and \Cref{fast-matrix-exp} to compute the distances to the $k$-nearest nodes, for a carefully chosen value of $k$.
	\item Compute a skeleton graph $G_S$ with $\tilde{O}(n/k)$ nodes such that an approximation of APSP on $G_S$ can be used to compute an approximation of APSP on by \Cref{lemma_skeleton_simplified}.
	\item Utilizing the fact that $G_S$ has only $\tilde{O}(n/k)$ nodes, compute an $O(\sqrt{a})$-approximation of APSP on $G_S$ in $O(1)$ rounds, which gives an $O(\sqrt{a})$-approximation of APSP on $G$.
\end{enumerate}

To complete the description of the algorithm, we should specify the value of $k$, and explain how we compute approximate APSP on $G_S$ in the last step. All other steps take $O(1)$ rounds. To specify the value of $k$, we focus first on the case of $a \in O(\log{n})$. Intuitively, we would like to choose $h \in \Theta\left(\sqrt{\log{n}}\right)$ and $k \in n^{\Theta\left(1/\sqrt{\log{n}}\right)}=2^{\Theta\left(\sqrt{\log{n}}\right)}$. It can be verified that with this choice of parameters, we can use \Cref{fast-matrix-exp} to compute the distances to the $k$-nearest nodes in $O(1)$ rounds. Now, if the number of nodes in $G_S$ is $\tilde{O}(n/k) = n/2^{\Theta(\sqrt{\log{n}})}$, then we can compute an $O(\sqrt{\log{n}})$-approximation for APSP in $G_S$ as follows. We construct an $O(\sqrt{\log{n}})$-spanner of $G_S$ with $O(n/k)^{1+\Theta\left(1/\sqrt{\log{n}}\right)}\subseteq O(n)$ edges. Such spanners can be constructed in $O(1)$ rounds~\cite{chechik2022constant}. Since the spanner has $O(n)$ edges, we can broadcast it to the whole graph and get an $O(\sqrt{\log{n}})$-approximation of the distances in $G_S$ and in $G$, as needed. 

Informally, if the weighted diameter constraint can be entirely ignored, then starting from an $O(\log n)$-approximation and applying the approximation factor reduction lemma iteratively for $O(\log \log \log n)$ rounds yields an $O(1)$-approximation for APSP.

More precisely, suppose the weighted diameter $d$ of the original graph $G$ satisfies $d \in (\log n)^{O(1)}$. In this case, the iterative approximation factor reduction approach can only improve the approximation ratio to $\left( \log \log n \right)^{O(1)}$. Beyond this point, the precondition of \Cref{approximate-reduction} is no longer satisfied, so we must switch to a different method. To further improve the approximation ratio from $\left( \log \log n \right)^{O(1)}$ to $O(1)$, we apply the $O(\log \log n)$-round approach discussed earlier. Since we now start with a $\left( \log \log n \right)^{O(1)}$-approximation, the number of hops in the resulting hopset is also $\left( \log \log n \right)^{O(1)}$. This allows us to use $i \in O(\log \log \log n)$ in the step that computes distances to the $\sqrt{n}$-nearest nodes. Consequently, the overall round complexity is $O(\log \log \log n)$, as desired.

\paragraph{APSP approximation in general graphs (\Cref{sect:APSPgeneral}).} 
The final challenge is to eliminate the assumption that the weighted diameter satisfies $\log d \in a^{O(1)}$. To overcome this, we prove a \emph{weight scaling} lemma that reduces distance approximation on the original graph to distance approximation on \( O(\log n) \) graphs \( G_1, G_2, \ldots, G_{O(\log n)} \), each with weighted diameter \( (\log n)^{O(1)} \). The reduction applies to any pair of nodes within \( (\log n)^{O(1)} \) hops. The key intuition is that distances that are exponential in $i$ are handled by the graph \( G_i \).

Due to the \( (\log n)^{O(1)} \)-hop constraint of our weight scaling lemma, we must first compute a hopset before applying it. However, the hopset only guarantees distance approximation between each node and its \( \sqrt{n} \)-nearest nodes. Consequently, even after computing \( O(1) \)-approximate APSP on the graphs \( G_1, G_2, \ldots, G_{O(\log n)} \), we can only hope to recover \( O(1) \)-approximate distances to the \( \sqrt{n} \)-nearest nodes in the original graph \( G \).
 
Recall from our \( O(\log \log n) \)-round \( O(1) \)-approximate APSP algorithm that, once we have the \emph{exact} distances to the \( \sqrt{n} \)-nearest nodes, we can construct a skeleton graph and use it to compute \( O(1) \)-approximate APSP. In the full version of \Cref{lemma_skeleton_simplified}, we show that a skeleton graph can still be built even from \emph{approximate} distances to the \( \sqrt{n} \)-nearest nodes, allowing us to obtain \( O(1) \)-approximate APSP in the end.

However, running an algorithm in parallel on the \( O(\log n) \) graphs \( G_1, G_2, \ldots, G_{O(\log n)} \) introduces an additional \( O(\log n) \) factor in bandwidth. To mitigate this, we once again leverage skeleton graphs. In \( O(1) \) rounds, we construct a skeleton graph \( G_S = (V_S, E_S) \) with \( |V_S| = n / \poly(\log n) \) nodes and reduce the APSP approximation problem on \( G \) to one on \( G_S \). This guarantees that each node in \( G_S \) has a bandwidth budget of \( |V_S| \cdot \poly(\log n) \).

\section{Fast Computation of \texorpdfstring{$\sqrt{n}$-Nearest  $\beta$-Hopsets}{sqrt(n)-Nearest beta-Hopsets}} \label{sec:hopset}

In this section, we prove \Cref{approximate-to-hop}.

\ahop*

We emphasize that, although this paper focuses on APSP in \emph{undirected} graphs, \Cref{approximate-to-hop} also holds for \emph{directed} graphs. The graph $G$ considered in this section is assumed to be directed.

\paragraph{Intuition behind the proof.} 
Our goal is to use an $a$-approximation for APSP to add shortcut edges such that, for each node $v$, all nodes in $N_{\sqrt{n}}(v)$ can be reached from $v$ within a small number of hops. In other words, we begin with \emph{approximate} distances and aim to construct an \emph{exact} hopset.

The algorithm proceeds as follows. For a given node $v$, we consider its approximate $\sqrt{n}$-nearest nodes, denoted $\tilde{N}_{\sqrt{n}}(v)$. Each node $u \in \tilde{N}_{\sqrt{n}}(v)$ sends to $v$ its $\sqrt{n}$ shortest outgoing edges. Using this information, node $v$ updates its approximate distances to nodes in $\tilde{N}_{\sqrt{n}}(v)$ and adds shortcut edges accordingly. We show that, for a significant fraction of the nodes in $\tilde{N}_{\sqrt{n}}(v)$, the resulting distances are \emph{exact}.

More concretely, we prove that for any node $u$ in the \emph{exact} set of $\sqrt{n}$-nearest nodes ${N}_{\sqrt{n}}(v)$, a shortcut edge combined with one edge from the original graph suffices to cover at least a $1/a$ fraction of the distance from $v$ to $u$. Repeating this process reduces the remaining distance by a $1/a$ factor in each step. Hence, node $u$ can be reached from $v$ within $O(a \log d)$ hops, where $d$ is the weighted diameter of $G$.

\subsection{Algorithm Description}

Let $\delta $ be a given approximation of APSP satisfying that for all $u,v \in V$, we have \[d(u,v) \leq \delta(u,v) \leq a \cdot d(u,v).\]
The algorithm for construction the hopset $H$ is as follows.
\begin{enumerate}
	\item Each node $v$ computes the set $\tilde{N}_{\sqrt{n}}(v)$ consisting of the $\sqrt{n}$ nodes $u$ with the smallest values of $\delta(v, u)$, breaking ties by $\ID$s.
	\item  Each node  $v$ asks each node $u  \in \tilde N_{\sqrt n}(v)$ for the $\sqrt n$ shortest outgoing edges from $u$ in $G$.\label{collect_edges}
	\item  Each node  $v$ runs a shortest-path algorithm on the subgraph of $G$ induced by all the edges $v$ received in the step above, together with all the outgoing edges from $v$ in $G$. \label{step:localdist}
	\item Each node $v$ adds an edge $e = (v, u)$ with weight $d'(v, u)$ to the hopset $H$, where $d'(v, u)$ is the distance from $v$ to $u$ computed by node $v$ in the step above. \label{step:addedges}
\end{enumerate}

We emphasize that in both Step~\ref{collect_edges} and Step~\ref{step:localdist}, the underlying distance metric is \( d \) and not \( \delta \). In Step~\ref{step:addedges}, although the edge \( e \) is initially known only to node \( v \) and not to node \( u \), we can ensure that each edge in the hopset is known to both endpoints by simply having \( v \) send the edge \( e \) to \( u \). This communication can be done in parallel for all such pairs \( (v, u) \) in a single round. It is possible that $d'(v,u) \neq d'(u,v)$ even if $G$ is undirected.

\begin{claim}\label{claim:round_complexity}
	The above algorithm can be implemented in $O(1)$ rounds of \clique.
\end{claim}
\begin{proof}
        The only part that requires communication is Step~\ref{collect_edges}. In this routing task, every node is the destination of $O(n)$ messages.
        At the start of the algorithm, each node holds $O(n \log n)$ bits of information. Before the routing task,  only one round has passed, due to sending the requests.
        Thus, we can apply \Cref{lemma:improved-routing} to complete the routing task in $O(1)$ rounds.
\end{proof}

\subsection{Correctness}
For each $v \in V$ and $d > 0$, we define the following:

\begin{itemize}
  \item $B_d(v) = \{ u \in V : d(v, u) \le d \}$, i.e., it is the closed ball of radius $d$ centered at $v$. 
  \item $\ell(v) = \max_{u \in N_{\sqrt{n}}(v)} d(v, u)$, i.e., it is the smallest integer $\ell$ such that $|B_\ell(v)| \ge  \sqrt{n}$.
\end{itemize}

By the definition of $\ell(v)$, for any $\ell' < \ell(v)$, we have $|B_{\ell'}(v)| < \sqrt{n}$.
Our first goal is to show that each node $v$ computes the exact distance to each node $u$ such that $d(v,u) \leq (\ell(v)-1)/a$. %
We make the following claim.

\begin{claim} \label{claim_Bv}
    For each $v \in V$, we have $B_{(\ell(v)-1)/a}(v) \subseteq  \tilde N_{\sqrt n}(v)$.
\end{claim}

\begin{proof}
    For all $u \in B_{(\ell(v)-1)/a}(v)$, we have $d(v, u) \leq (\ell(v)-1)/a$, and hence $\delta(v, u) \leq \ell(v)-1$. In addition, by the definition of $\ell(v)$, there are less than $\sqrt{n}$ nodes at distance at most $\ell(v)-1$ from $v$, which implies that there are less than $\sqrt{n}$ nodes $u$ with $\delta(v,u) \leq \ell(v)-1$, as $d(v,u) \leq \delta(v,u)$. Hence, all nodes $u$ with $d(v,u) \leq (\ell(v)-1)/a$ are necessarily in the set $\tilde N_{\sqrt n}(v)$.
\end{proof}
  
We prove that each node $v$ learns the exact distance to all nodes in a sufficiently small set. Recall from the algorithm description that $d'(v, u)$ is the distance from $v$ to $u$ locally computed by node $v$.

\begin{lemma} \label{lemma_Bv}
For all $u  \in  B_{(\ell(v)-1)/a}(v)$, we have $d'(v, u)=d(v, u)$.
\end{lemma}
\begin{proof}
Assume, for the sake of contradiction, that the claim is false. Let $u \in B_{(\ell(v)-1)/a}(v)$ be a counterexample with the smallest value of $d(v, u)$.
By the algorithm description, $d'(v, u) \geq d(v, u)$, so we have $d'(v, u) > d(v, u)$. Let $t$ be the node right before $u$ on some shortest path from $v$ to $u$, so $t \in B_{(\ell(v)-1)/a}(v)$, and our choice of $u$ implies that
\[
d'(v, t) = d(v, t).
\]
Since $d'(v, u) > d(v, u)$, the edge $(t, u)$ on that shortest path must not have been sent to $v$ during the execution of the algorithm.
By \Cref{claim_Bv}, $t \in \tilde{N}_{\sqrt{n}}(v)$, so $t$ must have sent to $v$ its shortest $\sqrt{n}$ outgoing edges. Let the other endpoints of these edges be $u_1, u_2, \ldots, u_{\sqrt{n}}$.
By the choice of these edges, for each $i\in \{1, \ldots, \sqrt{n}\}$,
\[
d(t, u_i) \leq d(t, u).
\]
Therefore,
\[
d(v, u_i) \leq d(v, u),
\]
which implies $u_i \in B_{(\ell(v)-1)/a}(v)$. Hence, $B_{(\ell(v)-1)/a}(v)$ contains all nodes $u_1, u_2, \ldots, u_{\sqrt{n}}$ and also $u$, meaning
\[
|B_{(\ell(v)-1)/a}(v)| > \sqrt{n},
\]
contradicting the fact that $|B_{(\ell(v)-1)/a}(v)| < \sqrt{n}$, since by definition $B_{(\ell(v)-1)/a}(v) = B_{\ell'}(v)$ for some $\ell' < \ell(v)$.
\end{proof}

Our next goal is to prove that there is a low-hop path from each node $v$ to each node $u \in N_k(v)$. Our proof exploits the following property of $\ell(v)$.

\begin{claim} \label{claim_lv}
    For two nodes $u$ and $v$, it holds that $\ell(v)-\ell(u) \leq  d(v,u)$.
\end{claim}
\begin{proof}
Suppose, for the sake of contradiction, that 
$\ell(u) + d(v,u) < \ell(v)$. 
On one hand, we have 
$|B_{\ell(u) + d(v,u)}(v)| < \sqrt{n}$,
since $\ell(u) + d(v,u) < \ell(v)$.
On the other hand, the ball $B_{\ell(u) + d(v,u)}(v)$ covers the set $B_{\ell(u)}(u)$, which has size at least $\sqrt{n}$. Hence
$|B_{\ell(u) + d(v,u)}(v)| \geq \sqrt{n}$,
which is a contradiction.
\end{proof}

To show that $H$ is indeed a $\sqrt{n}$-nearest $O(a \log{d})$-hopset, we prove something stronger: There is an $O(a \log d)$-hop path from $v$ to $u$ for any $u \in B_{\ell(v)}(v)$. Since there are at least $\sqrt{n}$ nodes at distance at most $\ell(v)$ from $v$, this implies that there is an $O(a \log d)$-hop path from $v$ to its $\sqrt{n}$-nearest nodes.

\begin{lemma}\label{lem:hopset_final}
	Let $v  \in  V$, and $u \in B_{\ell(v)}(v)$. 
 Then there is an $O(a \log d)$-hop path from $v$ to $u$ in $G \cup H$ with length $d(u, v)$.
\end{lemma}
\begin{proof}
Fix a shortest path $v=s_1 \rightarrow s_2 \rightarrow \cdots \rightarrow s_k=u$ in $G$, where $u  \in  B_{\ell(v)}(v)$. 

Let $t_{0}=s_{1}$. For each $i \in \mathbb {Z}^{+}$, let $t_{i}$ be the smallest-index node after $t_{i-1}$ on the path that does not belong to $B_{(\ell(t_{i-1})-1)/a}(t_{i-1})$. We terminate when $u \in B_{(\ell(t_i)-1)/a}(t_i)$. 
See \Cref{figure_ti} for illustration.

\begin{center}
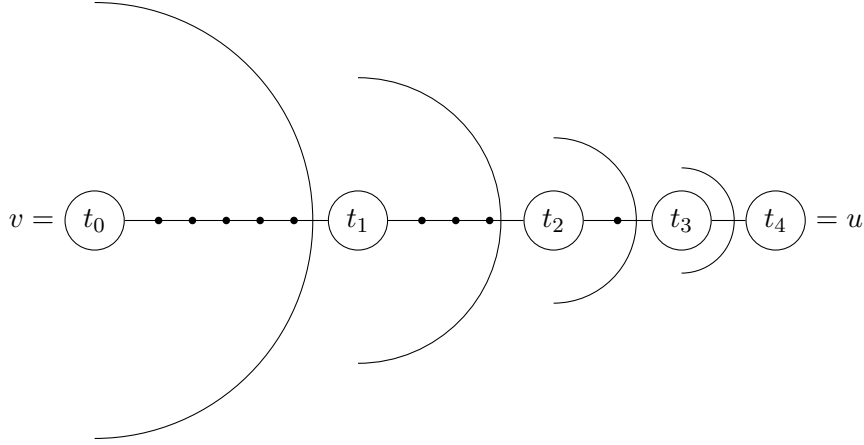

\begin{tikzpicture}[
	bignode/.style={circle, draw, minimum size=0.7cm},
	smallnode/.style={circle, fill, minimum size=0.1cm, inner sep=0pt},
	]

	\node [bignode](t0) {$t_0$};
	\path (t0.west) node [anchor=east] {$v=$};

	\draw (t0.east)
		-- ++(0.45, 0) node [smallnode] {}
		-- ++(0.45, 0) node [smallnode] {}
		-- ++(0.45, 0) node [smallnode] {}
		-- ++(0.45, 0) node [smallnode] {}
		-- ++(0.45, 0) node [smallnode] {}
		-- ++(0.45, 0) node [anchor=west, bignode](t1) {$t_1$}

		(t1.east)
		-- ++(0.45, 0) node [smallnode] {}
		-- ++(0.45, 0) node [smallnode] {}
		-- ++(0.45, 0) node [smallnode] {}
		-- ++(0.45, 0) node [anchor=west, bignode](t2) {$t_2$}

		(t2.east)
		-- ++(0.45, 0) node [smallnode] {}
		-- ++(0.45, 0) node [anchor=west, bignode](t3) {$t_3$}

		(t3.east)
		-- ++(0.45, 0) node [anchor=west, bignode](t4) {$t_4$}
		;
	
    \draw (t0) ++(270:2.9cm) arc (270:450:2.9cm);
    \draw (t1) ++(270:1.9cm) arc (270:450:1.9cm);
    \draw (t2) ++(270:1.1cm) arc (270:450:1.1cm);
    \draw (t3) ++(270:0.7cm) arc (270:450:0.7cm);
	\path (t4.east) node [anchor=west] {$=u$};
\end{tikzpicture}
\captionof{figure}{Illustration of the selection of $t_i$.}\label{figure_ti}
\end{center}

By an induction on $i$, the following can be proved:
\begin{itemize}
  \item We claim that $\ell(s_{i})\geq d(s_{i},u)$. For $s_1=v$ this clearly holds, as $u \in B_{\ell(v)}(v)$. For a general $i$, by \Cref{claim_lv}, we have $\ell(s_{i})\geq \ell(s_{1})-d(s_{1},s_{i})\geq d(s_{1},u)-d(s_{1},s_{i})=d(s_{i},u).$
  \item Our choice of $t_{i+1}$ guarantees $t_{i+1}\notin B_{(\ell(t_i)-1)/a}(t_i)$, which implies $d(t_{i+1},u)\leq d(t_{i},u)-\frac{\ell(t_{i})}{a}\leq d(t_{i},u)-\frac{d(t_{i},u)}{a}=(1-\frac{1}{a})\cdot d(t_{i},u)$.
\end{itemize}

Let $i^\ast$ be the final index of the sequence $(t_0, t_1, t_2, \ldots)$. We claim that $i^\ast \leq \lceil a \ln d \rceil +1$, since otherwise \[d(t_{\lceil a \ln d \rceil+1},u) \leq d(v,u)\cdot \left(1-\frac{1}{a}\right)^{\lceil a \ln d \rceil +1}  < d(v,u) \cdot \frac{1}{d} \leq 1,\]
meaning that $d(t_{\lceil a \ln d \rceil+1},u) = 0$ and $t_{\lceil a \ln d \rceil+1} = u$.

By \Cref{lemma_Bv}, our choice of the sequence $(t_0, t_1, t_2, \ldots)$ implies that $t_{i+1}$ can be reached from $t_i$ via a 2-hop path of length $d(t_{i-1},t_i)$ in $G \cup H$. Similarly,  if $u \neq t_{i^\ast}$, then $u$ can be reached from $t_{i^\ast}$ via one edge of weight $d(t_{i^\ast}, u)$ in $G \cup H$. Therefore, there is an $O(a \log{d})$-hop path from $v$ to $u$ with length $d(v,u)$ in $G \cup H$, as needed.
\end{proof}

We summarize the discussion as follows.

\begin{proof}[Proof of \Cref{approximate-to-hop}]
By \Cref{claim:round_complexity}, the algorithm finishes in $O(1)$ rounds.  By the algorithm description, the weight of every edge $e=(v,u)$ added to $H$ equals the length of a path from $v$ to $u$ in the input graph $G$. Therefore, the distances are the same in both $G$ and $G \cup H$.
Consider any two nodes $v \in V$ and $u \in N_{\sqrt{n}}(v) \subseteq B_{\ell(v)}(v)$.
By \Cref{lem:hopset_final}, there is a path from $v$ to $u$ in $G \cup H$ with length $d_{G}(v,u)$, so $H$ is indeed a $\sqrt{n}$-nearest $O(a \log d)$-hopset.
\end{proof}

\section{Fast Computation of the \texorpdfstring{$k$}{k}-Nearest Nodes} \label{sec:knearest}

In \Cref{sec:hopset}, we show that it is possible to construct a hopset such that, for every node $v$, its $k$-nearest nodes $N_k(v)$ are reachable within $O(a \log d)$ hops. Assuming that the edge weights are polynomially bounded and that an $a \in O(\log n)$-approximation of APSP is given, the resulting hopset achieves a hop bound of $O(\log^2 n)$.

Our next goal is to use this hopset to compute the distances to the $k$-nearest nodes. A naive approach would require $O(\log^2 n)$ rounds, as the hopset guarantees paths with $O(\log^2 n)$ hops. By applying fast matrix exponentiation, following the approach of~\cite{censorhillel2019fast}, the computation can be done in $O(\log(\log^2 n)) = O(\log \log n)$ rounds. The high-level idea is to compute, in round $i$, the $k$-nearest nodes with respect to the $2^i$-hop distances. 

Nevertheless, even this improved approach remains too costly for our purposes. To obtain a more efficient algorithm, we identify specific conditions under which the distances to the $k$-nearest nodes can be computed in just $O(1)$ rounds. We begin by proving the following lemma.

\begin{restatable}{lemma}{knearest}\label{lemma:fast-matrix-exp-one-iter}
Let $h$ and $k$ be positive integers satisfying $k \in O(n^{1/h})$. In a weighted directed graph $G$, we can deterministically compute, for each node $u$, the $h$-hop distances to the nodes $N_k^h(u)$, in $O(1)$ rounds.
\end{restatable}

The problem can be phrased as a special case of \emph{filtered matrix multiplication}. Here, we are given the weighted adjacency matrix $A$ of the input graph $G$, and our goal is to compute for each node $u$, the $k$ smallest elements, breaking ties by node $\ID$s, in the row of $u$ in the matrix $A^{h}$, where $A^{h}[u,v]$ equals the $h$-hop distance from $u$ to $v$. 

To obtain an efficient algorithm, we start by filtering the matrix $A$ to a matrix $\overline{A}$ that only keeps the $k$ smallest elements in each row, breaking ties by node $\ID$s, and then we compute $\overline{A}^{h}$. We will later show that  $\overline{A}^{h} = \overline{A^h}$.

Filtered matrix multiplication was also used before to find the $k$-nearest nodes in \cite{censorhillel2019fast,dory2022exponentially}. They showed how to compute filtered matrix multiplication of two matrices in $O(1)$ rounds when $k \leq \sqrt{n}$. They also showed an algorithm that works when $k \leq n^{2/3}$, but it comes at a cost of an additional $O(\log{n})$ term in the round complexity. To use their approach to compute the $k$-nearest nodes with respect to $h$-hop distances, one needs to repeat the algorithm $\log{h}$ times, which is costly.  

In this work, we give a much faster algorithm. Concretely, we show that $\overline{A}^{h} = \overline{A^h}$ can be computed in $O(1)$ rounds when $k \in O(n^{1/h})$.
Repeating this $i$ times yields the following result.

\begin{restatable}{lemma}{inearest}
	\label{lemma:fast-matrix-exp-i-iter}
    Let $h$, $k$, and $i$ be positive integers satisfying $k \in O(n^{1/h})$. In a weighted directed graph $G$, we can deterministically compute, for each node $u$, the $h^i$-hop distances to the nodes $N_k^{h^i}(u)$, in $O(i)$ rounds.
\end{restatable}

As a direct corollary of \Cref{lemma:fast-matrix-exp-i-iter}, if we have a $k$-nearest hopset that guarantees low-hop paths from each node to its $k$-nearest nodes, then we can compute the exact distances to these nodes using \Cref{lemma:fast-matrix-exp-i-iter}.

\matrixexp*
\begin{proof}
    It follows by applying \Cref{lemma:fast-matrix-exp-i-iter} to $G \cup H$.
\end{proof}

Similar to \Cref{approximate-to-hop}, although this paper focuses on APSP in \emph{undirected} graphs, the result in \Cref{fast-matrix-exp} also holds for \emph{directed} graphs. Throughout this section, we assume that the input graph $G$ is directed.

\subsection{Overview and Comparison With Existing Algorithms}

Our approach is inspired by previous algorithms for neighborhood collection~\cite{chechik2022constant} and path listing~\cite{CensorHillel2021FastDA} in sparse graphs. Specifically, the algorithm in~\cite{chechik2022constant} targets low-degree graphs, while the one in~\cite{CensorHillel2021FastDA} assumes the absence of small cycles. In contrast, our input graph is not necessarily sparse, so we cannot collect full neighborhoods as in those works. Instead, we leverage the fact that we are only interested in a sparse part of the output---namely, the $k$-nearest nodes---to design an efficient algorithm.

Our algorithm is a derandomization of the neighborhood collection algorithm in~\cite{chechik2022constant} using the bin-splitting technique from~\cite[Theorem 8]{CensorHillel2021FastDA}, which resembles the color-coding method of Alon, Yuster, and Zwick~\cite{alon1995color}. In addition, we use a more careful selection of edges to transmit and a tighter analysis to handle cases where not all nodes have low degree.

Compared to the path listing algorithm in~\cite[Theorem 8]{CensorHillel2021FastDA}, which enumerates all paths of $O(k)$ hops when the graph has at most $O(n^{1+1/k})$ edges, the first phase of our algorithm is identical. However, in our case, we additionally require the nodes to learn the path lengths. Specifically, if $v \in N_k^h(u)$, then node $u$ must learn the pair $(v, \min \ell_P)$, where $\ell_P$ is the length of a path $P$ from $u$ to $v$ with at most $h$ hops, and the minimum is taken over all such paths.

In the setting of~\cite{CensorHillel2021FastDA}, it is guaranteed that each node $u$ is the start of only $O(n)$ such paths $P$, so $u$ receives at most $O(n)$ messages of the form $(v, \ell_P)$. This allows for a direct application of \Cref{lemma:improved-routing}. In our setting, however, this guarantee no longer holds. To address this, the key idea is to sort the edge list in a particular order. We describe the algorithm below.

\subsection{Algorithm Description}
\label{subsec:fast-matrix-exp-algorithm}

	We now describe the algorithm used to prove \Cref{lemma:fast-matrix-exp-one-iter}.

	\begin{enumerate}
		\item \label{step_filter}
        Each node $u$ computes a list $M_{(u)}$ consisting of $k$ triplets $(u, v, w_{uv})$, where $v \in N_k^1(u)$ and $w_{uv}$ is the weight of the \emph{directed} edge $(u, v)$. That is, $u$ selects the $k$ nodes $v$ with the smallest edge weights $w_{uv}$. Each triplet represents an edge and its weight. We will later show that this information is sufficient for each node to reach its $k$-nearest nodes.

\item  \label{step_bins}
Let $M = M_{(1)} \circ M_{(2)} \circ \cdots \circ M_{(n)}$ be the \emph{ordered} list formed by concatenating the lists $M_{(v)}$ for all nodes $v$, in increasing order of their $\ID$s. Here, we slightly abuse notation by writing $M_{(\ID(v))} = M_{(v)}$. The list $M$ contains $nk \in O(n \cdot n^{1/h})$ elements.

We divide $M$ into $p = \left\lfloor n^{1/h} \cdot \frac{h}{4} \right\rfloor$ contiguous sublists, each containing $O(n/h)$ edges. Each sublist is called a \emph{bin}, and we denote the ordered sequence of bins by $C_1, \dots, C_p$.

We consider all $h \cdot \binom{p}{h}$ ways of selecting $h$ distinct bins $C_{i_1}, C_{i_2}, \dots, C_{i_h}$, where the ordering of all bins \emph{except the first one}---i.e., $C_{i_2}, \dots, C_{i_h}$---is irrelevant. Each such selection is referred to as an \emph{$h$-combination}. We will later show that the number of $h$-combinations is at most $n$.

\item We assign each $h$-combination to a distinct node and have each such node learn all the triplets contained in the bins of its assigned $h$-combination. \label{step_assign}

\item \label{step_end}
Each node $u$ identifies all nodes $v$ such that the first bin $C_{i_1}$ in the $h$-combination assigned to $v$ contains at least one edge from $M_{(u)}$. Then, $u$ queries each such node $v$ for the $k$ nodes that are nearest to $u$, breaking ties by node $\ID$s, computed using the edges received by $v$ along paths of at most $h$ hops, as well as their corresponding $h$-hop distances. 
\end{enumerate}

\paragraph{Assumptions.} For the rest of the discussion, we make the following two assumptions.
\begin{enumerate}
    \item $p \geq h$, so it is possible to select $h$ distinct bins out of $p$ bins.
    \item The bin size $nk/p \in O(n/h)$ is larger than the local list size $k$.
\end{enumerate}
To justify the two requirements, we show that the problem becomes trivial if any one of the two requirements is not met. If $p < h$, then $k \in O(n^{1/h})\subseteq O(1)$, so we can afford to let all nodes broadcast their $k$ shortest outgoing edges in $O(1)$ rounds. This provides enough information to achieve the goal of \Cref{lemma:fast-matrix-exp-one-iter}. If the bin size $nk/p \in O(n/h)$ at most the local list size $k$, then $n \leq p = \left\lfloor n^{1/h} \cdot \frac{h}{4} \right\rfloor$. Since $h$ is a positive integer, this is only possible when $k \in O(n^{1/h})\subseteq O(1)$. As discussed earlier, in this case, the goal of \Cref{lemma:fast-matrix-exp-one-iter} can be achieved in $O(1)$ rounds.

\subsection{Complexity}

We demonstrate a constant-round implementation of the above algorithm. The high-level idea is as follows. Given the values of $h$ and $p$, each node can locally compute the assignment of the $h$-combinations using a deterministic procedure.
In Step~\ref{step_assign}, each node needs to learn $O(n)$ words, each of $O(\log n)$ bits. Hence we can apply \Cref{lemma:improved-routing} to perform the communication in $O(1)$ rounds.
Similarly, in Step~\ref{step_end}, each node $u$ requests $O(k)$ words of information from $O(n/p)$ nodes $v$. Since $k \in O(n^{1/h})$ and $p = \lfloor n^{1/h} \cdot \frac{h}{4} \rfloor$, the total number of words received by any node remains $O(n)$. Therefore, we can again apply \Cref{lemma:improved-routing} to complete this step in $O(1)$ rounds.

\begin{lemma} \label{complexity_knearest}
	The above algorithm can be executed in $O(1)$ rounds.
\end{lemma}

\begin{proof} We now describe how the algorithm can be efficiently implemented in the \clique model. Step~\ref{step_filter} requires no communication. Now we consider Step~\ref{step_bins}. First of all, the number $h \cdot \binom{p}{h}$  of $h$-combinations is at most $n$ because
\[h \cdot \binom{p}{h} \leq  h \cdot \left(\frac{pe}{h}\right)^h \leq  h \cdot \left(n^{1/h} \cdot \frac{e}{4}\right)^h = n \cdot h \cdot \left(\frac{e}{4}\right)^h.\]
For all $h \geq 1$, $h \cdot \left(\frac{e}{4}\right)^h \leq  1$ because $\ln\left(h \cdot \left(\frac{e}{4}\right)^h \right) = \ln h + h \ln \frac{e}{4}$ has derivative $\frac{1}{h} + \ln \frac{e}{4}$. For $h \geq 3$, the derivative is negative, thus the function is decreasing on $h  \in [3, \infty )$, and it can be manually checked that $h \cdot \left(\frac{e}{4}\right)^h \leq  1$ for $h \in [1,3]$.

The list $M$ is not known to any individual node. However, since each local list $M_{(v)}$ has size exactly $k$, every node can use this information to determine, for any element $M_j$ of the global list $M$, which node owns it and its position within that node\IeC {\textquoteright }s local list. Specifically, given an index $j$, each node can compute the corresponding node $v$ and position $\ell$ such that $M_j$ is the $\ell$th element of $M_{(v)}$.

Given the values of $h$ and $p$, the nodes can use some deterministic algorithm to locally compute how the $h$-combinations should be assigned to the nodes, thus each node knows the assignment of $h$-combinations to all nodes.

\paragraph{Implementation of Step~\ref{step_assign}.} Consider any node $u$, which wants to learn the contents of bins $C_{i_1}, C_{i_2}, \dots, C_{i_h}$. As discussed earlier, node $u$ can locally compute the location of the edges that it needs to learn in each local list $M_{(v)}$.
Moreover, as the bin size $O(n/h)$ is  larger than the local list size $k$,  these edges must correspond to a contiguous segment of $M_{(v)}$. Let the leftmost index be $l_{uv}$ and the rightmost index be $r_{uv}$. 

For any two nodes $u$ and $v$, let $u$ send the values $l_{uv}$ and $r_{uv}$ to $v$. This allows each node to know exactly which portion of its list to send and to whom.

Each bin has $O(n \cdot n^{1/h} /p) \subseteq  O(n/h)$ edges, and each node
 has to learn $h$ bins, so each node need to receive $O(n)$ messages of $O(\log n)$ bits, so the routing algorithm of \Cref{lemma:improved-routing} can be applied to solve the routing instance in $O(1)$ rounds.

\paragraph{Implementation of Step~\ref{step_end}.} Let $S$ be the set of nodes $v$ such that the bin $C_{i_1}$ assigned to $v$ contains at least one edge from $M_{(u)}$. Recall that node $u$ knows the $h$-combination assignments for all nodes, as well as the contiguous sublist of $M$ corresponding to the edges in $M_{(u)}$. Using this information, it can compute the set $S$ locally. In this step, node $u$ needs to request the $k$ nodes nearest to $u$ within $h$ hops computed by each node $v \in S$.

Recall that there are a total of $h \cdot \binom{p}{h} \leq n$ possible $h$-combinations.
By symmetry, for each $1 \leq i \leq p$, the fraction of $h$-combinations where $i_1 = i$ is $1/p$.
Therefore, the number of such $h$-combinations is at most $n/p$.
Furthermore, since each bin is much larger than any local list, at most two bins can contain elements from $M_{(u)}$.
Hence the size of the set $S$ satisfies $|S| \leq 2n/p \in O(n/p)$.

Since $k \in O(n^{1/h})$ and $p = \left\lfloor n^{1/h} \cdot \frac{h}{4} \right\rfloor$, we have
\[
|S| \cdot k \in O(n/p) \cdot k \subseteq O(n).
\]
Therefore, using the routing algorithm of \Cref{lemma:improved-routing}, each node can obtain all the relevant information in $O(1)$ rounds.
\end{proof}

\subsection{Correctness} %

To prove \Cref{lemma:fast-matrix-exp-one-iter}, our goal is to show that each node $u$ learns the $h$-hop distances to its $k$-nearest nodes $N_k^h(u)$ with respect to $h$-hop distances. Intuitively, this holds because each such $h$-hop path is known to some node $v$, which will relay the necessary information to $u$. We show that filtering the edges to retain only the $k$ smallest-weight edges per node does not eliminate any of these paths. \Cref{lemma:fast-matrix-exp-i-iter} then follows by applying \Cref{lemma:fast-matrix-exp-one-iter} iteratively $i$ times.

Framing this as a matrix problem, each node $u$ needs to learn the $k$ smallest entries in row $u$ of the matrix $A^{h}$, where $A^{h}[u,v]$ represents the $h$-hop distance from $u$ to $v$. We begin by showing that $u$ can obtain the $k$ smallest entries in row $u$ of the matrix $\overline{A}^{h}$, where $\overline{A}$ is derived from $A$ by retaining only the $k$ smallest entries in each row, with ties broken by node $\ID$s, and setting all other entries to $\infty$.

Observe that $\overline{A}$ is the weighted adjacency matrix of the \emph{directed} graph defined by the local edge lists constructed in Step~\ref{step_filter} of the algorithm. Even when the initial graph $G$ is undirected, the graph becomes directed after filtering in the sense that an edge $e = \{u, v\}$ may be retained by one endpoint $u$ but discarded by the other endpoint $v$. In such a case, $e$ is treated as a directed edge $(u, v)$, which appears in row $u$ but not in row $v$ of $\overline{A}$.

\begin{lemma} \label{lemma_Abar}
    By the end of the algorithm, each node $u$ knows the distances to the $k$-nearest nodes according to $\overline{A}^{h}$. That is, $u$ learns the $k$ pairs $\left(w, \overline{A}^{h}[u,w]\right)$ corresponding to the $k$ smallest elements in the row of $u$ in $\overline{A}^{h}$, breaking ties by node $\ID$s.
\end{lemma}

\begin{proof}
By definition, $\overline{A}^{h}[u, w]$ is the minimum length of a path $P$ with at most $h$ hops from $u$ to $w$ in $\overline{A}$. Since $\overline{A}$ is the weighted adjacency matrix obtained after filtering  (Step~\ref{step_filter}), each edge of such a path $P$ belongs to some bin, and there exists an $h$-combination $(C_{i_1}, C_{i_2}, \dots, C_{i_h})$ such that each edge of $P$ lies in some bin in the $h$-combination, with the first edge in $C_{i_1}$. Thus, the entire path $P$ is known to the node $v$ assigned to the $h$-combination $(C_{i_1}, C_{i_2}, \dots, C_{i_h})$.

In Step~\ref{step_end} of the algorithm, node $v$ sends to node $u$ the $k$ nodes that are closest to $u$\IeC {\textemdash }with ties broken by node $\ID$s\IeC {\textemdash }along with their corresponding $h$-hop distances, as computed from the edges received by $v$ along paths of at most $h$ hops.

If $\overline{A}^{h}[u, w]$ is among the $k$ smallest entries in row $u$ of $\overline{A}^{h}$, with ties broken by node $\ID$s, then node $v$ must send the pair $\left(w, \overline{A}^{h}[u, w]\right)$ to $u$, since otherwise the $k$ nodes reported by $v$ would all be strictly preferred over $w$ in the selection of the $k$ smallest entries of row $u$, which contradicts the assumption that $\overline{A}^{h}[u, w]$ is among the $k$ smallest entries in row $u$ of $\overline{A}^{h}$.
\end{proof}

We next show that learning the $k$ smallest elements of each row of according to $\overline{A}^h$ is equivalent to learning the $k$ smallest elements of each row of $A^h$. The following lemma states a property of filtered matrix multiplication that holds for all $i \geq 1$, and is independent of the algorithm being used.

\begin{lemma} \label{cor_filter}
For any integer $i \geq 1$, $\overline{\overline{A}^i} = \overline{A^i}$.
\end{lemma}
\begin{proof}
To prove the lemma, it suffices to show that for each $v \in N^i_k(u)$, any minimum-length path $P$ from $u$ to $v$ among all $u$-$v$ paths with at most $i$ hops in the original graph (whose weighted adjacency matrix is $A$) is also preserved in the graph resulting from keeping only the $k$ smallest-weight outgoing edges for each node (whose weighted adjacency matrix is $\overline{A}$).

Let $d_P(u, w)$ be the distance from $u$ to a node $w$ along the path $P$. Consider the node $w$ right before $v$ in $P$, and examine the subpath from $u$ to $w$. Since all edge weights are positive, we have $d_P(u, x) < d_P(u, v)$ for every node $x$ on this subpath.

Suppose, for contradiction, that some edge on this subpath is removed in $\overline{A}$. Let $(x, y)$ be the \emph{first} such edge in $P$. Since $(x, y)$ is filtered out, node $x$ must have retained $k$ edges $(x,z)$ such that $w_{xz} \leq w_{xy}$. For each such $z$, the path from $u$ to $z$ via $x$ has length at most 
\[
d_P(u, x) + w_{xy} = d_P(u, y) < d_P(u, v),
\]
and uses fewer than $i$ hops. This contradicts the assumption that $v$ is among the $k$-nearest nodes to $u$ using at most $i$ hops.

Thus, the entire subpath from $u$ to $w$ is preserved. A similar argument applies to the final edge $(w, v)$. If $(w, v)$ was removed, then $w$ would have retained $k$ neighbors strictly closer than $v$, breaking ties by node $\ID$s, again contradicting the assumption that $v \in N^i_k(u)$.

Therefore, the path $P$ from $u$ to $v$ is fully preserved in $\overline{A}$. Since this holds for every $v \in N^i_k(u)$, the $k$ smallest entries in row $u$ of $\overline{A}^i$ are identical to those in $A^i$. Applying this argument to all nodes $u$, we conclude that $\overline{\overline{A}^i} = \overline{A^i}$, completing the proof.
\end{proof}

We can now prove \Cref{lemma:fast-matrix-exp-one-iter}.

\knearest*

\begin{proof}
   By \Cref{lemma_Abar}, by the end of the algorithm, each node $u$ knows the distances to the $k$-nearest nodes according to the matrix $\overline{A}^{h}$. Hence, by \Cref{cor_filter}, each node $u$ knows the distances to the $k$-nearest nodes according to $A^h$, which means that $u$ knows the $h$-hop distances to the $k$-nearest nodes $N_k^h(u)$ with respect to $h$-hop distances.
 The number of rounds is $O(1)$ by \Cref{complexity_knearest}. 
\end{proof}

By repeating the algorithm $i$ times, we prove \Cref{lemma:fast-matrix-exp-i-iter}.

\inearest*

\begin{proof}
    Let $A$ be the weighted adjacency matrix of $G$. From \Cref{lemma:fast-matrix-exp-one-iter}, if we apply the algorithm once, then each node $u$ learns the distances to the $k$-nearest nodes according to $A^h$. We denote the resulting matrix by $A_1 = \overline{A^h}$. We can now run the algorithm again with the input matrix $A_1$, and so on. We prove that after $i$ iterations, each node $u$ knows the distances to the $k$-nearest nodes according to $A^{h^i}$, which are exactly the $h^i$-hop distances to the nodes $N_k^{h_i}(u)$. 
    
    The proof is by an induction on $i$. For $i=1$, it holds by \Cref{lemma:fast-matrix-exp-one-iter}. Assume that it holds for $i$, and we prove that it holds for $i+1$. This means that after $i$ iterations, each node knows the distances to the $k$-nearest nodes according to $A^{h^i}$. We write $A_i= \overline{A^{h^i}}$ to denote the result of filtering $A^{h^i}$. We now run the algorithm on $A_i$. By \Cref{lemma:fast-matrix-exp-one-iter}, we know that after running the algorithm, each node knows distances to the $k$-nearest nodes according to $\left(A_i\right)^h = \left(\overline{A^{h^i}}\right)^h$. By \Cref{cor_filter}, $\overline{\left(\overline{A^{h^i}}\right)^h} = \overline{A^{h^{i+1}}}$, so   
    this is equivalent to knowing the  distances to the $k$-nearest nodes according to $A^{h^{i+1}}$, as needed. The number of rounds is $O(i)$, as each iteration takes $O(1)$ rounds. 
\end{proof}

\section{Skeleton Graphs}
\label{sec:skeleton}

After computing distances to the $k$-nearest nodes, our next goal is to extend it to approximate distances between all pairs of nodes. 
Informally, the following lemma states that once each node knows its approximate $k$-nearest nodes, we can construct a graph $G_S$ with $\tilde{O}(\frac{n}{k})$ nodes, such that approximating APSP on $G$ reduces to approximating APSP on $G_S$. For technical reasons, during our algorithms sometimes nodes only know the distances to an approximate set of the $k$-nearest nodes. To handle this, we require certain conditions on the sets known, which makes the full description of the lemma a bit more technical. 

\begin{lemma}[Skeleton graphs -- full version]\label{skeleton} Let \( a \geq 1 \) be a real number, and let \( k \geq 1 \) be an integer. Let \( \delta \) be a symmetric function that takes any two nodes \( u \) and \( v \) in a weighted undirected graph \( G \) as input and returns an integer. Suppose that for each node \( u\) in $G$, a set \( \tilde{N}_k(u) \) of \( k \) nodes is given, satisfying the following conditions:
\begin{itemize}
    \item For each node \( v \in \tilde{N}_k(u) \), the value \( \delta(u, v) \) is known to \( u \), and satisfies
    \[
    d(u,v) \leq \delta(u, v) \leq a \cdot d(u, v).
    \]
    \item For any \( v \in \tilde{N}_k(u) \) and \( t \notin \tilde{N}_k(u) \), we have
    \[
    \delta(u, v) \leq a \cdot d(u, t).
    \]
\end{itemize}
Then, w.h.p., we can construct a weighted undirected graph \( G_S \) over a subset \( V_S \subseteq V \) of size \( O\left( \frac{n \log k}{k} \right) \) in \( O(1) \) rounds. Moreover, given any \( l \)-approximation of APSP on \( G_S \), we can deterministically compute a \( 7 l a^2 \)-approximation of APSP on \( G \) in \( O(1) \) rounds.
\end{lemma}

Informally, the set $\tilde{N}_k(u)$ is understood as an $a$-approximate $k$-nearest set of nodes for $u$, and $\delta $ is a \emph{local} $a$-approximation of APSP on $G$ in the sense that the $a$-approximation guarantee only applies to the approximate $k$-nearest sets $\tilde N_k(u)$. Unlike the $k$-nearest set $N_k(u)$ which is unambiguously defined, $\tilde N_k(u)$ can be any set held by node $u$ that satisfies the conditions listed in the lemma statement, and such a set is not necessarily unique.

Essentially, this lemma allows us to extend a local $a$-approximation $\delta $ over the approximate $k$-nearest sets to a $7la^2$-approximation of APSP on the whole of $G$, losing a multiplicative factor of $7la$ in the approximation.
The following lemma is a special case of \Cref{skeleton}.

\skel*
\begin{proof}
It follows from \Cref{skeleton} by setting  $a=1$, $\delta (u, v)=d(u, v)$, and $\tilde{N}_k(u) = N_k(u)$.
\end{proof}

For all applications of skeleton graphs in this work, except for the proof of \Cref{thm:special-case-2-approximate-APSP,thm:special-case-2-approximate-APSP-truncated}, each node $u$ already knows its exact set $N_k(u)$ of its $k$-nearest nodes, so the above simplified version suffices.

\subsection{Algorithm Description} \label{sec:alg_skeleton}
We describe the algorithm for \Cref{skeleton}.

\begin{enumerate}
	\item First, we construct a hitting set \( S \subseteq V \) of size \( O\left( \frac{n \log k}{k} \right) \) such that for every node \( v \), we have \( S \cap \tilde{N}_k(v) \neq \emptyset \). This can be done in \( O(1) \) rounds w.h.p. The nodes in \( S \) are referred to as the \emph{skeleton nodes}.

\item For each node \( u \), we select \( c(u) \in S \cap \tilde{N}_k(u) \) to be the skeleton node closest to \( u \) with respect to \( \delta \), i.e., minimizing \( \delta(u, c(u)) \), breaking ties using node $\ID$s.

\item We construct the skeleton graph \( G_S \) as follows. The node set of \( G_S \) is \( S \). The edge set of \( G_S \) is defined by the following rule: For any triplet of nodes \( u, v, t \) satisfying:
\begin{itemize}
    \item \( t \in \tilde{N}_k(u) \),
    \item \( \{t, v\} \in E \) or \( t = v \),
\end{itemize}
we add an edge between \( c(u) \) and \( c(v) \) to \( G_S \), with weight:
\[
    \delta(c(u), u) + \delta(u, t) + w_{tv} + \delta(v, c(v)),
\]
where \( w_{tv} \) is the weight of edge \( \{t, v\} \).

\item Suppose we are given a function \( \delta_{G_S} \) that provides an \( l \)-approximation to APSP on \( G_S \). We compute a function \( \eta \) that gives a \( 7 l a^2 \)-approximation to APSP on \( G \) as follows:
\begin{itemize}
    \item If \( u \in \tilde{N}_k(v) \) or \( v \in \tilde{N}_k(u) \), set \( \eta(u, v) = \delta(u, v) \).
    \item Otherwise, set 
    $\eta(u, v) = \delta(u, c(u)) + \delta_{G_S}(c(u), c(v)) + \delta(c(v), v)$.
\end{itemize}
\end{enumerate}

Intuitively, the above procedure defines a clustering, where \( S \) is the set of cluster centers, and each node \( u \) is assigned to its closest center \( c(u) \). The construction of edges in the skeleton graph can be viewed as a form of two-hop exploration: Starting from \( u \), move to a node \( t \in \tilde{N}_k(u) \), and from there take one additional hop via an edge in \( E \). For any node \( v \) reachable by such a process, we add an edge between \( c(u) \) and \( c(v) \). A similar two-hop exploration also appears in the proof of \Cref{lem:hopset_final}. 

As usual, in the presence of parallel edges, only the one with the minimum weight is retained. We emphasize that the function $\eta$ is symmetric, i.e., $\eta(u, v) = \eta(v, u)$.

\subsection{Complexity}

We show that the above algorithm can be implemented in \( O(1) \) rounds. The key ideas are outlined below. The selection of the hitting set \( S \) is done via a randomized algorithm: Each node includes itself in \( S \) independently with probability \( \Theta\left( \frac{|S|}{n} \right) \). To ensure that \( S \) is a desired hitting set, an additional step to fix $S$ is needed.

The construction of the edge set of \( G_S \) is carried out using sparse matrix multiplication to compute 4-hop paths of the form \( c(u) \rightarrow u \rightarrow t \rightarrow v \rightarrow c(v) \). Careful ordering of the matrix multiplications is required to preserve sparsity at each step.
Similarly, computing \( \eta(u, v) \) involves sparse matrix multiplication to compute 3-hop paths of the form \( c(u) \rightarrow u \rightarrow v \rightarrow c(v) \).

The formal proof and details are presented below. We begin by introducing the sparse matrix multiplication algorithm that we rely on.

\paragraph{Sparse matrix multiplication algorithm.}  
For an \( n \times n \) matrix \( M \), define the \emph{density} \( \rho_M \) to be the average number of entries per row that are not equal to \( \infty \), which is the identity element of the \( \oplus \) operation. Equivalently, \( \rho_M \) is the total number of non-\( \infty \) entries in \( M \) divided by \( n \).

The following sparse matrix multiplication algorithm is taken from \cite{censorhillel2019fast}. It is particularly useful due to the well-known connection between matrix multiplication over semirings and the computation of shortest paths.

In the matrix multiplication problem, the goal is to compute the product \( ST \) of two given \( n \times n \) matrices \( S \) and \( T \) over the min-plus semiring.

\begin{theorem}[{\cite[Theorem 8]{censorhillel2019fast}}]
	\label{thm:fast-matrix-multiplication}
	The product \( ST \) of two given \( n \times n \) matrices \( S \) and \( T \) over the min-plus semiring can be computed deterministically in
	\[
		O\left( \frac{(\rho _S \rho _T \rho _{ST})^{1/3}}{n^{2/3}} + 1 \right)
	\]
	rounds in \clique, assuming that $\rho _{ST}$ is known beforehand.
\end{theorem}

We emphasize that \Cref{thm:fast-matrix-multiplication} remains applicable even when \( \rho_{ST} \) is only a known \emph{upper bound} on the average number of non-\( \infty \) entries per row.
Using \Cref{thm:fast-matrix-multiplication}, we show that our algorithm can be implemented in \( O(1) \) rounds.

\begin{lemma}\label{clm:skel-1}
The skeleton graph \( G_S \) over a subset \( V_S \subseteq V \) of size \( O\left( \frac{n \log k}{k} \right) \) can be constructed in $O(1)$ rounds w.h.p.
\end{lemma}
\begin{proof} We show that both the hitting set and the skeleton graph can be constructed in $O(1)$ rounds.

\paragraph{Hitting set.} 
We build the hitting set \(S\) using the method of \cite[Lemma 4.1]{dory2021constant}:

\begin{enumerate}
  \item Each node joins \(S\) independently with probability \(\frac{\ln k}{k}\).
  \item Any node \(v\) that still finds no member of \(S\) in \(\tilde{N}_k(v)\) also joins \(S\). \label{ss2}
\end{enumerate}
A node is selected in Step~\ref{ss2} with probability at most
\(\left(1-\tfrac{\ln k}{k}\right)^{k}\le \tfrac1k\).  
Thus, we still have \(\mathbb{E}[|S|] \in \Theta\left(\tfrac{n\ln k}{k}\right)\) by linearity of expectations. By Markov inequality,
\[\Pr\left[|S| \in O(\tfrac{n\ln k}{k})\right]\ge\frac12.\] 
Repeating the procedure \(O(\log n)\) times and selecting the smallest $S$ amplifies the probability to \(1-1/\poly(n)\).  
Each repetition uses only \(O(1)\) bits of communication between each pair of nodes, so all \(O(\log n)\) instances can run in parallel. 
To select the smallest $S$, we can designate $O(\log n)$ nodes $v_1, v_2, \ldots v_{O(\log n)}$, where $v_i$ is responsible for calculating the size of $S$ in the $i$th repetition, by letting all nodes send to $v_i$ an one-bit indicator for whether the node is in $S$. After that, all $v_i$ broadcast the size of $S$ in the $i$th repetition to the entire network.
Step~\ref{ss2} guarantees \(S\cap\tilde{N}_k(v)\neq\emptyset\) for every \(v\in V\), so \(S\) is a valid hitting set.

\paragraph{Skeleton Graph.}
Consider two skeleton nodes \( s_a, s_b \in S \). Although there may be multiple candidate edges between \( s_a \) and \( s_b \), we are only interested in the one with the minimum weight.

The desired edge corresponds to the minimum value of
\[
\delta(s_a, u) + \delta(u, t) + w_{tv} + \delta(v, s_b),
\]
over all nodes \( u, v, t \) such that \( c(u) = s_a \), \( c(v) = s_b \), and \( t \in \tilde{N}_k(u) \).

To compute this value efficiently, we define two functions. For each pair \( (s_a, t) \), let
\[
x(s_a, t) := \min_{u :\;\substack{c(u) = s_a \\ t \in \tilde{N}_k(u)}} \left( \delta(s_a, u) + \delta(u, t) \right).
\]
Similarly, for each pair \( (t, s_b) \), define
\[
y(t, s_b) := \min_{v :\; c(v) = s_b} \left( w_{tv} + \delta(s_b, v) \right).
\]
Then, the edge weight between \( s_a \) and \( s_b \) is given by
\[
\min_{t} \left( x(s_a, t) + y(t, s_b) \right).
\]

\paragraph{Computing the $x$-values and $y$-values.} We can compute the values \( x(s_a, t) \) and make them known to both \( s_a \) and \( t \) in \( O(1) \) rounds as follows. Each node \( u \) sends the tuple \( (c(u),\, \delta(c(u), u) + \delta(u, t)) \) to every node \( t \in \tilde{N}_k(u) \). Then, for each skeleton node \( s_a \), every node \( t \) selects the minimum second value among all received messages of the form \( (s_a,\, \cdot) \), thereby computing \( x(s_a, t) \). Finally, each \( t \) sends the corresponding value \( x(s_a, t) \) back to \( s_a \).

A similar procedure is used to compute \( y(t, s_b) \) in \( O(1) \) rounds. Each node \( v \) sends the tuple \( (c(v),\, w_{tv} + \delta(v, c(v))) \) to all of its neighbors \( t \) in the original graph. Then, for each skeleton node \( s_b \), each node \( t \) selects the minimum second value among all received messages of the form \( (s_b,\, \cdot) \), thereby computing \( y(t, s_b) \). Finally, each \( t \) sends \( y(t, s_b) \) to \( s_b \).

\paragraph{Matrix multiplication.} Finally, we obtain the edge weights of \( G_S \) via a single sparse min-plus matrix multiplication.  
Define two \( n \times n \) matrices \( X \) and \( Y \), where \( X[s_a, t] = x(s_a, t) \) and \( Y[t, s_b] = y(t, s_b) \) whenever these values are defined, and \( \infty \) otherwise. For any pair of skeleton nodes \( s_a, s_b \in S \), the weight of the edge \( \{s_a, s_b\} \) is given by
\[
\min_{t} \left( X[s_a, t] + Y[t, s_b] \right),
\]
which equals the \((s_a, s_b)\)-entry of the matrix product \( X Y \) over the min-plus semiring.

We now bound the density parameters required by \Cref{thm:fast-matrix-multiplication}.  
Matrix \( X \) has density \( \rho_X \leq k \), since a non-\( \infty \) entry  \( X[s_a, t] \) implies the existence of some node \( u \) such that \( c(u) = s_a \) and \( t \in \tilde{N}_k(u) \), and there are at most \( n k \) such pairs overall.  
Matrix \( Y \) has density \( \rho_Y \leq |S| \in O\left( \frac{n \log k}{k} \right) \), as any non-\( \infty \) entry  \( Y[t, s_b] \) satisfies $s_b \in S$, implying that each row has at most $|S|$ non-\( \infty \) entries.  
Similarly, the product \( XY \) satisfies \( \rho_{XY} \leq |S|^2/n \in O\left( \frac{n \log^2 k}{k^2} \right) \), since  \( XY \)  contains at most $|S|^2$ non-\( \infty \) entries.
Plugging these bounds into \Cref{thm:fast-matrix-multiplication} shows that the matrix multiplication completes in \( O(1) \) rounds.  
\end{proof}

\begin{lemma}\label{clm:skel-2}
 Given an $l$-approximation of APSP $\delta _{G_S}$ on $G_S$, the APSP approximation $\eta $ on $G$ can be computed in $O(1)$ rounds.
\end{lemma}
\begin{proof}
We show that all values \( \eta(u, v) \) can be computed efficiently. If \( u \in \tilde{N}_k(v) \) or \( v \in \tilde{N}_k(u) \), then the computation of \( \eta(u, v) = \delta(u, v) \) is trivial. Otherwise,
\[
\eta(u, v) = \delta(u, c(u)) + \delta_{G_S}(c(u), c(v)) + \delta(c(v), v).
\]
To express this computation in matrix form, define an \( |S| \times n \) matrix \( A \) such that \( A[c(u), u] = \delta(c(u), u) \) and all other entries are \( \infty \). Since each column contains exactly one finite entry, the density of \( A \) is \( \rho_A = 1 \). Let \( D \) be the \( |S| \times |S| \) matrix representing the approximate distances on the skeleton graph, where \( D[s_a, s_b] = \delta_{G_S}(s_a, s_b) \) and all other entries are \( \infty \).

With these definitions, computing \( \eta \) reduces to evaluating the matrix product \( A^\top D A \) over the min-plus semiring. Each multiplication involving \( A \) or \( A^\top \) takes \( O(1) \) rounds by \Cref{thm:fast-matrix-multiplication}, since their density is constant. Therefore, we can first compute the intermediate matrix \( B = D A \), and then compute \( A^\top B \) to obtain \( \eta \). The entire computation completes in \( O(1) \) rounds.
\end{proof}

\subsection{Correctness}

Consider any pair of nodes \( u \) and \( v \) in \( G \). We need to show that \( \eta(u, v) \leq 7 l a^2 \cdot d(u, v) \). If \( v \in \tilde{N}_k(u) \) or \( u \in \tilde{N}_k(v) \), then we have \( \eta(u, v) = \delta(u, v) \leq a \cdot d(u, v) \), and the claim holds. Therefore, we may assume that \( v \notin \tilde{N}_k(u) \) and \( u \notin \tilde{N}_k(v) \) in the subsequent discussion.

\paragraph{Decomposition into segments.}

Let \( u = v_0 \to v_1 \to \cdots \to v_m = v \) be a shortest path from \( u \) to \( v \).  
We define a sequence of nodes \( u_0, t_0, u_1, t_1, \dots, u_p, t_p \) along this path as follows:

\begin{itemize}
    \item Set \( u_0 = u \).
    \item For each \( i \geq 0 \), let \( t_i \) be the rightmost node on the path (i.e., the node \( v_j \) with the largest index \( j \)) such that \( t_i \in \tilde{N}_k(u_i) \).
    \item Define \( u_{i+1} \) to be the node immediately following \( t_i \) on the path.
    \item Repeat this process until the last selected node \( t_p \) equals \( v \).
\end{itemize}

An illustration of this construction is shown in \Cref{fig:skeleton-illustration}.
Since \( v \notin \tilde{N}_k(u) \), it follows that \( p \geq 1 \). For notational convenience, we define \( s_i = c(u_i) \) for each \( 0 \leq i \leq p \), and \( s_\ast = c(v) \). While it may happen that \( u_i = t_i \), we always have \( t_i \neq u_{i+1} \) by construction.

	\begin{figure}
		\def\nodedistancea{0.95*2cm}
		\def\nodedistanceb{0.95*1.3cm}
		\def\nodedistancec{0.95*1.5cm}
		\begin{center}
\begin{tikzpicture}[scale=1, node distance=\nodedistancea, bignode/.style={circle, draw, minimum size=0.7cm}, red/.style={fill=red!30}]
	\node [bignode] (u0) {$u_0$};
	\path (u0.west) node [anchor=east] {$u=$};
	\node [bignode, red](s0) [above of=u0, node distance=\nodedistancea] {$s_0$};
	\node [bignode] (w0) [right of=u0] {$t_0$};
	\node [bignode] (u1) [right of=w0, node distance=\nodedistanceb]{$u_1$};
	\node [bignode, red](s1) [above right of=u1, node distance=\nodedistancea] {$s_1$};
	\node [bignode] (w1) [right of=u1] {$t_1$};
	\node [bignode] (u2)  [right of=w1, node distance=\nodedistanceb]{$u_2$};
	\node [bignode, red](s2) [above right of=u2, node distance=\nodedistancec] {$s_2$};
	\node [bignode](w2) [right of=u2, node distance=\nodedistancec] {$t_2$};
	\node [bignode, red](sv) [below of=w2, node distance=\nodedistancec] {$s_*$};
	\path (w2.east) node [anchor=west] {$=v$};

    \draw (u0)--(w0);
	\draw [dashed] (u0)--(s0);
    \draw (w0)--(u1);
	\draw [dashed] (u1)--(s1);
    \draw (u1)--(w1);
    \draw (w1)--(u2);
	\draw [dashed] (u2)--(s2);
    \draw (u2)--(w2);
	\draw [dashed] (w2)--(sv);

    \draw (u0) circle (2.5cm);
    \draw (u1) ++(-90:2.5cm) arc (-90:90:2.5cm);
    \draw (u2) ++(-90:2.7cm) arc (-90:90:2.7cm);
\end{tikzpicture}
\end{center}
\caption{Illustration of the construction. Red nodes are skeleton nodes.}
\label{fig:skeleton-illustration}
\end{figure}
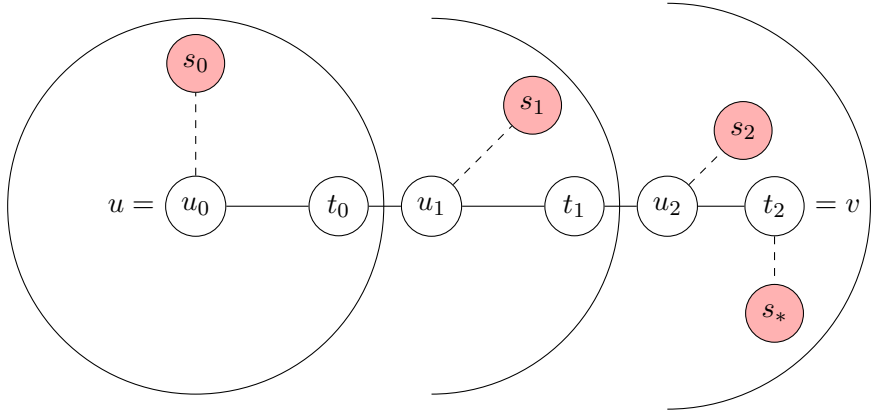

\paragraph{The special case of \( a = 1 \).}
As a warm-up before presenting the full proof, we consider the special case when \( a = 1 \), meaning \( \delta(u, v) = d(u, v) \) is the exact distance between \( u \) and \( v \). We sketch the argument that, in this setting, the \( \eta \)-distance along the path 
\[
u=u_0 \dashrightarrow s_0 \to s_1 \to \cdots \to s_p \to s_\ast \dashrightarrow t_p = v
\]
is at most a constant multiple of \( d(u, v) \).

By the definition of $\eta$, the weight of each edge \( s_i \to s_{i+1} \) under $\eta$ can be upper bounded by the length the path under $\delta=d$
\[
s_i \dashrightarrow u_i \to t_i \to u_{i+1} \dashrightarrow s_{i+1},
\]
where the solid arrows (\( \to \)) represent edges that lie on the original shortest path from \( u \) to \( v \), and the dashed arrows (\( \dashrightarrow \)) represent edges outside that path. The dashed arrows connect a node to its cluster center in the skeleton graph construction. 
Similarly, for the last segment \( s_p \to s_{*} \), we can consider the path
\[
s_p \dashrightarrow u_p \to t_p \dashrightarrow s_{*}.
\]
Our proof uses a charging argument. For each dashed arrow in the constructed path, we identify a subpath of the original shortest path from \( u \) to \( v \) whose length is at least that of the dashed arrow. We do this in such a way that each edge of the shortest path is charged at most \( O(1) \) times, which implies the desired \( O(1) \)-approximation.

\begin{itemize}
    \item For each \( i < p \), the segment \( u_i \dashrightarrow s_i \) and its reversal \( s_i \dashrightarrow u_i \) are charged to the subpath \( u_i \to u_{i+1} \) of the original shortest path. Since \( s_i \in \tilde{N}_k(u_i) \) and \( u_{i+1} \notin \tilde{N}_k(u_i) \), we have \(d(u_i, s_i) \leq d(u_i, u_{i+1}) \).
    
    \item The length of the segment \( u_p \dashrightarrow s_p \), as well as its reversal \( s_p \dashrightarrow u_p \), is at most the length of the segment \( u_p \to t_p \dashrightarrow s_* \), by the choice of \( s_p \).
    
    \item Finally, the segment \( t_p \dashrightarrow s_* \), as well as its reversal \( s_* \dashrightarrow t_p \), has length at most \( d(u, v) \), so we can charge it to the entire $u$-$v$ path.
\end{itemize}

\paragraph{The general case.}

We start by proving inequalities that we use to bound $\eta (u, v)$.

\begin{claim}\label{ineq:a1}
    For each $i \in \{0,1,\ldots, p\}$, $d_{G_S}(s_i,s_{i+1}) \leq \delta (s_i, u_i)+\delta (u_i, t_i)+d(t_i, u_{i+1})+\delta (u_{i+1}, s_{i+1})$.
\end{claim}

\begin{proof}
Since $t_i$ and $u_{i+1}$ are adjacent on the shortest path from $u$ to $v$, we have $d(t_i, u_{i+1}) = w_{t_i u_{i+1}}$. The claim then follows directly from the construction of $G_S$, which adds an edge between $s_i = c(u_i)$ and $s_{i+1} = c(u_{i+1})$ with weight 
$\delta(s_i, u_i) + \delta(u_i, t_i) + w_{t_i u_{i+1}} + \delta(u_{i+1}, s_{i+1})$.
\end{proof}

\begin{claim}\label{ineq:a2}
 $d_{G_S} (s_p, s_*) \leq  \delta (s_p,u_p)+\delta (u_p, v)+\delta (v, s_*)$.
\end{claim}
\begin{proof}
Similar to the proof of \Cref{ineq:a1}, in the construction of $G_S$, we add an edge from $s_p$ to $s_*$ with weight $\delta(s_p, u_p) + \delta(u_p, t_p) + w_{t_p v} + \delta(v, s_*)$. Since $t_p = v$, we have $w_{t_p v} = 0$, and the claim follows.
\end{proof}

\begin{claim}\label{ineq:a3}
    For each $i \in \{0,1,\ldots, p\}$, $\delta (u_i, s_i) \leq  a \cdot d(u_i, u_{i+1})$.
\end{claim}
\begin{proof}
    This holds because, by their definitions, $s_i  \in  \tilde N_k(u_i)$ and $u_{i+1} \notin \tilde N_k(u_i)$.
\end{proof}

\begin{claim}\label{ineq:a4}
 $\delta (u_p, s_p) \leq  a \cdot d(u_p, s_*)$.
\end{claim}
\begin{proof}
Either $s_*$ is not in the set $\tilde{N}_k(u_p)$, in which case the claim follows as in the previous argument, or $s_* \in \tilde{N}_k(u_p)$. In the latter case, since $u_p$ chooses $s_p$ to be its center from within $\tilde{N}_k(u_p)$, we have
$\delta(u_p, s_p) \leq d(u_p, s_*) \leq a \cdot d(u_p, s_*)$.
\end{proof}

We are now ready to bound $\eta (u, v)$.

\begin{lemma}\label{lem:eta-bound}
    $\eta $ is a $7 l a^2$-approximation of APSP on $G$.
\end{lemma}
\begin{proof}
Consider any pair of nodes \( u \) and \( v \) in \( G \). If \( v \in \tilde{N}_k(u) \) or \( u \in \tilde{N}_k(v) \), then by the definition of \( \eta \) and the approximation guarantee of \( \delta \) in \Cref{skeleton}, we have
\[
d(u, v) \leq \delta(u, v) = \eta(u, v) \leq a \cdot d(u, v),
\]
as required.

For the remainder of the proof, assume that \( v \notin \tilde{N}_k(u) \) and \( u \notin \tilde{N}_k(v) \). In this case, we have
\[
\eta(u, v) = \delta(u, c(u)) + \delta_{G_S}(c(u), c(v)) + \delta(c(v), v).
\]
By the approximation guarantees of \( \delta \) and \( \delta_{G_S} \) in \Cref{skeleton} and our construction of \( G_S \), we have
\[
\eta(u, v) \geq d(u, c(u)) + d(c(u), c(v)) + d(c(v), v) \geq d(u, v).
\] To complete the proof, it remains to show that \( \eta(u, v) \leq 7 l a^2 \cdot d(u, v) \).
\colorlet{darkgreen}{green!70!black}
\begin{align*}
\eta (u, v) 
&=  \delta (u, s_0)+\delta (v, s_*)+\delta _{G_S} (s_0, s_*) \\
&\leq  l \cdot ( \delta (u, s_0)+\delta (v, s_*)+d_{G_S} (s_0, s_*) ) \\
&\leq  l \cdot \bigg( \delta (u, s_0)+\delta (v, s_*) \\
&\qquad + \sum _{i=0}^{p-1} \mathcolor{red}{d_{G_S} (s_i, s_{i+1})} + \mathcolor{blue}{d_{G_S} (s_p, s_*)} \bigg) \quad \text{(triangle inequality for $d_{G_S}$)} \\
&\leq  l \cdot \bigg( \delta (u, s_0)+\delta (v, s_*) \\
& \qquad + \sum _{i=0}^{p-1} \mathcolor{red}{\delta (u_i, s_i)+\delta (u_i, t_i)+d(t_i, u_{i+1})+\delta (u_{i+1}, s_{i+1})} \\
&\qquad +\mathcolor{blue}{\delta (u_p,s_p)+\delta (u_p, v)+\delta (v, s_*)} \bigg) \quad \text{(from \textcolor{red}{\Cref{ineq:a1}} and \textcolor{blue}{\Cref{ineq:a2}})}\\
&= l \cdot \bigg( 
\sum _{i=0}^{p-1} 2 \mathcolor{magenta}{\delta (u_i, s_i)}+\delta (u_i, t_i)+d(t_i, u_{i+1}) \\
&\qquad +2 \mathcolor{orange}{\delta (u_p,s_p)}+\delta (u_p, v)+2 \delta (v, s_*) \bigg) \\
&\leq  l \cdot \bigg( 
	\sum _{i=0}^{p-1} 2 \mathcolor{magenta}{a \cdot d(u_i, u_{i+1})}+\delta (u_i, t_i)+d(t_i, u_{i+1}) \\
&\qquad +2\mathcolor{orange}{a \cdot d(u_p,s_*)}+\delta (u_p, v)+2 \delta (v, s_*) \bigg) \quad \text{(from \textcolor{magenta}{\Cref{ineq:a3}} and \textcolor{orange}{\Cref{ineq:a4}})}\\
&\leq  l \cdot \bigg( 
\sum _{i=0}^{p-1} 2 a \cdot d(u_i, u_{i+1})+\mathcolor{darkgreen}{a \cdot d(u_i, t_i)}+d(t_i, u_{i+1}) \\
&\qquad +2a \cdot d(u_p,s_*)+\mathcolor{darkgreen}{a \cdot d(u_p, v)}+2\mathcolor{darkgreen}{a \cdot d(v, s_*)} \bigg) \quad \text{(\textcolor{darkgreen}{assumption on $\delta $})} \\
&\leq  la \cdot \bigg( 
	\sum _{i=0}^{p-1} 2 d(u_i, u_{i+1})+d(u_i, t_i)+d(t_i, u_{i+1}) \\
	&\qquad +2d(u_p,s_*)+d(u_p, v)+2d(v, s_*) \bigg) \quad \text{($a \geq 1$)} \\
&= la \cdot \left( 3 d(u_0, u_p) + 2d(u_p,s_*)+d(u_p, v)+2d(v, s_*) \right) \\
&\leq  la \cdot \left( 3 d(u_0, u_p) + 2\left(d(u_p,v)+d(v, s_*)\right)+d(u_p, v)+2d(v, s_*) \right) \\
& \hspace{7cm} \text{(triangle inequality for $d$)} \\
&= la \cdot \left( 3 d(u, v) + 4 d(v, s_*) \right) \\
&\leq  la \cdot \left( 3 d(u, v) + 4 \delta (v, s_*) \right) \\
&\leq  la \cdot \left( 3 d(u, v) + 4a \cdot d(v, u) \right) \\
&= la \cdot (3 + 4a) \cdot d(u, v) \\
&\leq  7la^2 \cdot d(u, v). \qedhere
\end{align*}
\end{proof}

We summarize the discussion as a proof of \Cref{skeleton}.

\begin{proof}[Proof of \Cref{skeleton}]
By \Cref{clm:skel-1}, we can construct a weighted undirected graph \( G_S \) over a subset \( V_S \subseteq V \) of size \( O\left( \frac{n \log k}{k} \right) \) in \( O(1) \) rounds w.h.p. Given any \( l \)-approximation of APSP on \( G_S \), we can deterministically compute a \( 7 l a^2 \)-approximation of APSP on \( G \) in \( O(1) \) rounds. The round complexity follows from \Cref{clm:skel-2}, and the approximation factor from \Cref{lem:eta-bound}.
\end{proof}

\section{APSP Approximation in Small Weighted Diameter Graphs}\label{sec:algos}

In this section, we combine the ingredients developed so far to obtain $O(\log \log \log n)$-round algorithms for $O(1)$-approximation of APSP, assuming the weighted diameter is bounded by $d \in (\log n)^{O(1)}$. Recall that $\clique[B]$ denotes the variant of the \clique model where messages have size $O(B)$.

\begin{theorem}[restate=approximateAPSPweak, name=APSP approximation in small weighted dimater graphs]\label{weaker-approximate-APSP-first}\label{weaker-approximate-APSP-2}
	Suppose the weighted diameter of the graph $G$ is $d  \in  (\log n)^{O(1)}$. The following statements hold.
    \begin{itemize}
        \item In the \clique model, in $O(\log \log \log n)$ rounds, a $21$-approximation of APSP can be computed w.h.p.
        \item In the $\clique[\log^3 n]$ model, in $O(\log \log \log n)$ rounds, a $7$-approximation of APSP can be computed w.h.p.
    \end{itemize}
\end{theorem}

In \Cref{sec:append_prelim_algorithm}, we review an existing spanner algorithm~\cite{chechik2022constant}, which yields an $O(\log n)$-approximation of APSP in $O(1)$ rounds. In \Cref{subsec:algo:overview}, we prove \Cref{approximate-reduction}, which shows how to transform an $a$-approximation into an $O(\sqrt{a})$-approximation in $O(1)$ rounds. This reduction allows us to progressively improve the approximation factor from $O(\log n)$ to $O(\sqrt{\log n})$, $O(\log^{1/4} n)$, $O(\log^{1/8} n)$, and so on. This iterative improvement leads to the proof of \Cref{weaker-approximate-APSP-first} in \Cref{subsec:weaker-approximate-APSP}.

\subsection{Spanners} \label{sec:append_prelim_algorithm}

We use the following algorithm from \cite{chechik2022constant} to compute spanners. Among other uses, it is used to compute an $O(\log n)$-spanner that uses $O(n)$ edges, which is helpful to bootstrap our algorithm.
\begin{lemma}[{\cite[Theorems 1.2 and 1.3]{chechik2022constant}}]
	\label{constant-round-spanner-0}
	For any constant $\varepsilon > 0$ and integer $k \geq 1$, there is a constant-round  algorithm that w.h.p.\  computes the following spanners for weighted undirected graphs in the $\clique$ model.
    \begin{itemize}
        \item  A $(1+\varepsilon)(2k - 1)$-spanner with $O\left(n^{1 + 1/k}\right)$ edges.
        \item  A $(2k - 1)$-spanner with $O\left(k \cdot n^{1 + 1/k}\right)$ edges.        
    \end{itemize}
\end{lemma}

From the spanner-computation algorithm in \Cref{constant-round-spanner-0}, we have the following corollaries,
which are used to compute approximations to the shortest distance. 
In our algorithms we sometimes compute APSP on a smaller skeleton graph $G_S$ over $O(n^{1-1/b})$ nodes, such that the nodes of $G_S$ know their incident edges.

\begin{corollary}\label{constant-round-APSP}
Let \( b \geq 1 \) be an integer. Let \( G \) be the input graph with \( n \) nodes, and let \( G_S \) be a weighted undirected subgraph over a subset of \( N \in O(n^{1 - 1/b}) \) nodes. Assume that each node in \( G_S \) knows its incident edges in \( G_S \).
Then, for any constant \( \varepsilon > 0 \), there exists a constant-round \clique algorithm that, w.h.p., computes a \( (1+\varepsilon)(2b - 1) \)-approximation of APSP in \( G_S \). Moreover, the output is made known to all nodes in \( G \). %
\end{corollary}

\begin{proof}
	We can use \Cref{constant-round-spanner-0} to compute a $(1+\varepsilon )(2b-1)$-spanner with $O(N^{1+1/b})\subseteq O(n)$ edges for the graph $G_S$. Since the spanner has $O(n)$ edges, we can broadcast it to all nodes of $G$ in $O(1)$ rounds, which results in $(1+\varepsilon )(2b-1)$-approximation of all distances in $G_S$. %
\end{proof}

In particular, in the special case that $G_S = G$, we obtain an $(\alpha \log n)$-approximation for any constant $\alpha > 0$.

\begin{corollary}\label{constant-round-APSP-log-n}
For any constant $\alpha > 0$, an $(\alpha \log n)$-approximation of APSP in weighted undirected graphs can be computed w.h.p.\ in $O(1)$ rounds in the \clique model. Moreover, the output is known to all nodes in $G$.
\end{corollary}
\begin{proof}
This follows by setting the parameters in \Cref{constant-round-APSP} as $G_S = G$, $b = \left\lfloor \frac{\alpha \log n}{3} \right\rfloor$, and $\varepsilon = 0.1$, so that the approximation factor becomes $(1+\varepsilon)(2b - 1) \le \alpha \log n$ and $n^{1 - 1/b} \in \Theta(n)$. 
\end{proof}

\subsection{Approximation Factor Reduction}\label{subsec:algo:overview} %

We are now ready to prove \Cref{approximate-reduction}. 

\approximatereduction*

We now describe the high-level idea of the algorithm.

\begin{enumerate}
    \item We begin by applying \Cref{approximate-to-hop} to compute a $\sqrt{n}$-nearest \( O(a \log d) \)-hopset \( H \).
    
    \item Next, we use \Cref{fast-matrix-exp} to let each node learn its distances to the \( k \)-nearest nodes, for an appropriately chosen \( k \). We set \( h = \frac{1}{2} \cdot a^{1/4} \) and  \( k = n^{1/h} = n^{2a^{-1/4}} \). With these parameters, we can apply \Cref{fast-matrix-exp} to compute the distances from each node to its \( k \)-nearest nodes.
    
    \item We then apply \Cref{lemma_skeleton_simplified} to construct a skeleton graph \( G_S \) over a node set \( V_S \subseteq V \) of size \( O\left(\frac{n \log k}{k}\right) \), such that an approximate APSP  in \( G_S \) yields an approximate APSP in \( G \).
    
    \item Finally, leveraging the fact that \( G_S \) has only \( O\left(\frac{n \log k}{k}\right) \) nodes, we use \Cref{constant-round-APSP} to compute an \( O(\sqrt{a}) \)-approximation of APSP in \( G_S \), which in turn gives an \( O(\sqrt{a}) \)-approximation of APSP in \( G \).
\end{enumerate}

We next provide the full proof.

\begin{proofof}{\Cref{approximate-reduction}}
In this proof, we assume that \( n \) is sufficiently large so that \( \log(2\log n) \leq \sqrt{\log n} \) and \( a \leq \log^2 n \); otherwise, the problem can be solved by brute force in \( O(1) \) rounds.
Starting from an \( a \)-approximation of APSP, we apply the following four tools in sequence: \Cref{approximate-to-hop}, \Cref{fast-matrix-exp}, \Cref{lemma_skeleton_simplified}, and \Cref{constant-round-APSP}, as outlined in the high-level overview above.

\paragraph{Step 1: Hopset construction.}
We apply \Cref{approximate-to-hop} with the parameter \( a \) to obtain a \( \sqrt{n} \)-nearest \( \hat{h} \in  O(a \log d) \)-hopset \( H \).  
Since \( \log d \in a^{O(1)} \), it follows that \( \hat{h} \in a^{O(1)} \) as well.

\paragraph{Step 2: Computing the \( k \)-nearest nodes.}
Since \( H \) is a \( \sqrt{n} \)-nearest \( \hat{h} \)-hopset, it is also a \( k \)-nearest \( \hat{h} \)-hopset for any \( k \leq \sqrt{n} \). Our goal is to choose \( k \) so that we can apply \Cref{fast-matrix-exp}.
To do so, we must find parameters \( h \), \( k \), and \( i \) for \Cref{fast-matrix-exp} such that \( \hat{h} \leq h^i \) and \( k \in O(n^{1/h}) \). Since we aim for an \( O(1) \)-round algorithm, it is necessary that \( i \in O(1) \).

As shown earlier, \( \hat{h} \in a^{O(1)} \), so choosing \( h \in a^{\Omega(1)} \) guarantees that a constant \( i \) suffices. We set \( h = \frac{1}{2} \cdot a^{1/4} \) and \( k = n^{1/h} = n^{2a^{-1/4}} \).
The choice of the constant \( 1/4 \) may seem arbitrary for now, but we will see in later steps that it suffices.
With these parameters, \Cref{fast-matrix-exp} applies, and each node learns the distances to its \( k = n^{2a^{-1/4}} \)-nearest nodes.

\paragraph{Step 3: Skeleton graph construction.}
We now apply \Cref{lemma_skeleton_simplified}. The choice of parameters is again determined by the previous step: Each node knows the \emph{exact} distances to its \( k = n^{2a^{-1/4}} \)-nearest nodes, allowing us to invoke \Cref{lemma_skeleton_simplified} with this value of \( k \).
As a result, we construct a skeleton graph \( G_S \) over a node set \( V_S \subseteq V \) of size \( O\left(\frac{n \log k}{k}\right) \), such that approximate APSP in \( G_S \) leads to an approximate APSP in the original graph \( G \).

\paragraph{Step 4: APSP approximation.}
Finally, we apply \Cref{constant-round-APSP}. The choice of the parameter \( b \) is again constrained by the previous step: Since the skeleton graph \( G_S \) has \( O\left(\frac{n \log k}{k}\right) \) nodes, we must select \( b \) so that \( O\left(\frac{n \log k}{k}\right) \subseteq O\left(n^{1 - 1/b}\right) \). We will later show that setting \( b = \sqrt{a} \) satisfies this condition.

With this choice, \Cref{constant-round-APSP} produces a \( (1 + \varepsilon)(2b - 1) \)-approximation of APSP on \( G_S \). Then, by \Cref{lemma_skeleton_simplified}, this yields a \( 7(1 + \varepsilon)(2b - 1) \)-approximation of APSP on the original graph \( G \).
Choosing \( \varepsilon = \frac{1}{14} \), we obtain the bound \( 7(1 + \varepsilon)(2b - 1) < 15b = 15 \sqrt{a} \), as desired.

\paragraph{Choice of parameter.}
To complete the proof, it remains to show that selecting \( b = \sqrt{a} \) ensures
that \( O\left(\frac{n \log k}{k}\right) \subseteq O\left(n^{1 - 1/b}\right) \),
where \( k = n^{2a^{-1/4}} \).

Intuitively, ignoring the \( O(\log k) \) factor in \( O\left(\frac{n \log k}{k}\right) \) simplifies the expression to
\( \frac{n}{k} = n^{1 - 1/h} \),
where \( h = \frac{1}{2} \cdot a^{1/4} \). Thus, choosing \( b \) to be as small as \( h \) would suffice. The extra safety margin in our actual choice \( b = \sqrt{a} \) compensates for the \( O(\log k) \) factor, which is asymptotically much smaller than both \( n \) and \( k \).

More formally, we claim the following: For sufficiently large \( n \), and any real \( a \geq 1 \), setting \( k = n^{2a^{-1/4}} \) and \( b = \sqrt{a} \) guarantees that \( \frac{n \log k}{k} \leq n^{1 - 1/b} \).

By the assumption that $n$ is large enough, we have:
\begin{align*}
	&&a^{1/4} \log (2 \log n) &\leq  \log n \\
	\Longrightarrow && a^{1/4} \log \left( 2a^{-1/4}\log n\right) &\leq  \log n \\
	\Longrightarrow && a^{-1/4} &\geq \frac{\log \left( 2a^{-1/4}\log n \right)}{\log n}.  \\
	\intertext{Because $a \geq 1$, we have $a^{-1/4}\geq a^{-1/2}=\frac{1}{b}$. Combining estimates, we have:}
	&& 2a^{-1/4} - \frac{1}{b} &\geq \frac{\log \left( 2a^{-1/4}\log n \right)}{\log n} \\
	\Longrightarrow && 1 + \frac{\log \left( 2a^{-1/4}\log n \right)}{\log n} - 2a^{-1/4}  &\leq   1-\frac{1}{b}  \\
	\Longrightarrow && \log n + \log \left( 2a^{-1/4}\log n \right) - 2a^{-1/4}\log n  &\leq  \log n \cdot \left(1-\frac{1}{b}\right) \\
	\Longrightarrow && \log n + \log \log k - \log k &\leq  \log n \cdot \left(1-\frac{1}{b}\right) \\
	\Longrightarrow && n \cdot \frac{\log k}{k} &\leq  n^{1-1/b}.\tag*{\qedhere}
\end{align*}
\end{proofof}

\subsection{APSP Approximation}
\label{subsec:weaker-approximate-APSP}

We are now ready to prove \Cref{weaker-approximate-APSP-first}.

\approximateAPSPweak*

\begin{proof}
	Using \Cref{constant-round-APSP-log-n}, we can compute
	an \( a \)-approximation of APSP with \( a \in O(\log n) \) in $O(1)$ rounds.
We then repeatedly apply \Cref{approximate-reduction}. After 1, 2, 3, \dots\ iterations, we obtain progressively better approximations: \( O(\sqrt{\log n}) \), \( O(\log^{1/4} n) \), \( O(\log^{1/8} n) \), and so on. Thus, after \( O(\log \log \log n) \) rounds, we obtain an \( a \)-approximation with \( a \in O(\log \log n) \). At this point, due to the requirement \( \log d \in a^{O(1)} \), we cannot continue using \Cref{approximate-reduction} to further reduce the approximation factor.

	Instead, we apply a more direct approach, closely resembling the method used in \Cref{sec:tech_overview} to obtain an \( O(\log \log n) \)-round algorithm. First, we pass the \( a \)-approximation of APSP to \Cref{approximate-to-hop} to construct a \( \sqrt{n} \)-nearest \( O(a \log d) \)-hopset \( H \). Since \( a \in O(\log \log n) \) and \( \log d \in O(\log \log n) \), the hop bound is \( O(\log^2 \log n) \).

	Next, we apply \Cref{fast-matrix-exp} with parameters \( k = \sqrt{n} \), \( h = 2 \), and \( i \in O(\log \log \log n) \), such that \( h^i \in O(\log^2 \log n) \). This enables each node to learn the exact distances to its \( \sqrt{n} \)-nearest nodes in \( O(i) \subseteq O(\log \log \log n) \) rounds.

	We then apply \Cref{lemma_skeleton_simplified} with \( k = \sqrt{n} \). The resulting skeleton graph \( H \) contains \( O(\sqrt{n} \log n) \) nodes. Using \Cref{constant-round-spanner-0}, we compute a 3-spanner of this skeleton graph with \( O((\sqrt{n} \log n)^{1 + 1/2}) \subseteq O(n) \) edges, which can be broadcast to the entire network in \( O(1) \) rounds. Therefore, by the conclusion of \Cref{lemma_skeleton_simplified}, we obtain a 21-approximation of APSP in the original graph.

	If we are working in the \( \clique[\log^3 n] \) model, then all \( O(n \log^2 n) \) edges of the skeleton graph \( H \) can be broadcast and exact distances computed in \( O(1) \) rounds, yielding a 7-approximation of APSP in the original graph \( G \).
\end{proof}

\section{APSP Approximation in General Graphs}\label{sect:APSPgeneral}

In this section, we extend \Cref{weaker-approximate-APSP-first} to general graphs. To do so, in \Cref{subsect:scaling}, we prove a \emph{weight scaling} lemma that reduces distance approximation on the original graph to distance approximation on \( O(\log n) \) graphs, each with weighted diameter \( (\log n)^{O(1)} \). The reduction applies to any pair of nodes within \( (\log n)^{O(1)} \) hops. Using this weight scaling lemma, we prove the following result in \Cref{proof-special-case-2}.

\begin{theorem}[name={APSP approximation with large bandwidth}, restate=specialcaseii]
	\label{thm:special-case-2-approximate-APSP}
For any constant $\varepsilon > 0$, a $(7^3  +\varepsilon)$-approximation of APSP in weighted undirected graphs can be computed w.h.p.\ in $O(\log \log \log n)$ rounds in the $\clique[\log^4 n]$ model.
\end{theorem}
 
The bandwidth \( B = \log^4 n \) in \Cref{thm:special-case-2-approximate-APSP} arises from running the algorithm in \Cref{weaker-approximate-APSP-first}---which operates in the \( \clique[\log^3 n] \) model---on \( O(\log n) \) graphs given by the weight scaling lemma. In \Cref{subsec:general_case}, we further extend this result to the standard \clique model by using skeleton graphs to reduce the number of nodes by a polylogarithmic factor.

 \main*

In \Cref{subsec:truncated-algorithm}, by restricting the algorithm to \( O(t) \) rounds, we obtain the following tradeoff between round complexity and approximation factor.

 \truncatedalgorithm*

\subsection{Weight Scaling Lemma}\label{subsect:scaling}

We now state and prove the weight scaling lemma.

\begin{restatable}[Weight scaling lemma]{lemma}{LemmaDiam}\label{hop-to-diameter}
	Let $\delta $ be a given $h$-approximation of APSP for the weighted undirected graph $G$, and let $\varepsilon > 0$ be a constant.
	 We can deterministically compute in zero rounds a sequence of
	$O(\log n)$ weighted undirected graphs \[ \{ G_0, G_1, \ldots , G_{O(\log n)}\} \]
	over the node set of $G$ and with weighted diameter at most $\ceil{\frac{2}{\varepsilon }} \cdot h^2$, such that if an
	$l$-approximation of APSP on each of these graphs, $\{ \delta _{G_0}, \delta _{G_1}, \ldots , \delta _{G_{O(\log n)}}\}$, is computed, then we can deterministically compute in zero rounds a function $\eta $ satisfying the following conditions:
	\begin{itemize}
		\item For all pairs of nodes $u$ and $v$, \[\eta (u, v) \geq d_G(u, v).\]
		\item If  $u$ and $v$ have some shortest path between them of no more than $h$ hops in $G$, then \[\eta (u, v) < (1+\varepsilon )l \cdot d_G(u, v).\]
	\end{itemize}
\end{restatable}

In \Cref{hop-to-diameter}, as usual, the graphs \( G_i \) and the distance approximations \( \delta_{G_i} \) are known in the following sense: Each node \( u \) knows all edges incident to it in \( G_i \), as well as the value \( \delta_{G_i}(u, v) \) for every node \( v \).

In applications of \Cref{hop-to-diameter} in this paper, we have \( l < h \). Thus, we begin with a coarse \( h \)-approximation \( \delta \) and compute a more accurate \( (1 + \varepsilon)l \)-approximation \( \eta \) that applies to all node pairs \( (u, v) \) for which there exists a shortest path of at most \( h \) hops in \( G \).

In particular, if every pair of nodes in \( G \) is connected by a shortest path of at most \( h \) hops, then \( \eta \) is a \( (1 + \varepsilon)l \)-approximation of APSP.

\paragraph{Proof strategy.} Before presenting the full proof of \Cref{hop-to-diameter}, we briefly discuss the main ideas.
Consider the following $4$-hop path.
\tikzset{
  every point/.style = {circle, inner sep={3\pgflinewidth}, opacity=1, fill, solid},
  point/.style={insert path={node[every point, #1]{}}}, point/.default={},
  point name/.style = {insert path={coordinate (#1)}},
}
\begin{center}
\begin{tikzpicture}
	\draw (0, 0) [point]
		to node [auto] {$3$} ++(2, 0) [point]
		to node [auto] {$1$} ++(2, 0) [point]
		to node [auto] {$203$} ++(2, 0) [point]
		to node [auto] {$7$} ++(2, 0) [point]
		;
\end{tikzpicture}
\end{center}
Round up every edge weight to the nearest integer multiple of $x=10$.
\begin{center}
\begin{tikzpicture}
	\draw (0, 0) [point]
		to node [auto] {$10$} ++(2, 0) [point]
		to node [auto] {$10$} ++(2, 0) [point]
		to node [auto] {$210$} ++(2, 0) [point]
		to node [auto] {$10$} ++(2, 0) [point]
		;
\end{tikzpicture}
\end{center}

Even though each \emph{individual} edge weight may increase by a large factor---for example, the edge with original weight 1 is incremented by \( +900\% \)---the length of the path increases only modestly, from 214 to 240, corresponding to an overall increase of approximately \( +12\% \).

The total increment is bounded by \( x \cdot h \), where \( h\) is the number of hops. As long as \( x \cdot h\) remains small compared to the original path length \( \ell \), the overall distortion is limited. To ensure that the additive error does not exceed an \( \varepsilon \)-fraction of the original path length, it suffices to set \( x = \varepsilon \ell / h \).

After rounding, all edge weights become integer multiples of \( x \), so we may divide all weights by \( x \). With the choice \( x = \varepsilon \ell / h \), the path length becomes \( O(h / \varepsilon) \), which is within an \( O(1 / \varepsilon) \) factor of the original hop bound \( h \). Therefore, for any given path length \( \ell \), there exists a suitable choice of \( x \) such that the rounding and scaling process produces a short path  while guaranteeing a \( (1 + \varepsilon) \)-approximation to the original length.

In view of the above, we partition the set of all node pairs \( (u, v) \) into \( O(\log n) \) groups based on their distance \( d(u, v) \); specifically, the \( i \)th group contains those pairs for which \( d(u, v) \approx \frac{2^i}{\varepsilon} \). 
For each index \( i \), we construct a graph \( G_i \) by rounding each edge weight up to the nearest multiple of \( 2^i \), and then dividing all weights by \( 2^i \).
To ensure that the weighted diameter of \( G_i \) is small, we add an edge between every pair of nodes with an appropriate weight.
We then \emph{truncate} all edge weights to ensure that the weighted diameter of \( G_i \) remains small.
The initial distance approximation \( \delta \) is used to estimate the appropriate scale \( i \) for each node pair \( (u, v) \).

\begin{proof}[Proof of \Cref{hop-to-diameter}]  We begin with the construction of the graphs $\{ G_i \}$.
	We set $B=\ceil{\frac{2}{\varepsilon }}$.

 \paragraph{Three sequences of graphs: $\{H_i\}$, $\{K_i\}$, and $\{G_i\}$.} 
			Let $i \in O(\log n)$ be a non-negative integer. Define $x=2^i$. Construct the graphs $H_i$, $K_i$, and $G_i$ as follows:
 \begin{description}
    \item[Construction of \( H_i \):]
    Let \( H_i \) be the graph resulting from rounding each edge weight in \( G \) up to the nearest integer multiple of \( x \); that is, for each edge \( \{u, v\} \), set
    \[
    w^{H_i}_{uv} = \left\lceil \frac{w_{uv}^G}{x} \right\rceil \cdot x.
    \]
    Observe that \( w_{uv}^G \leq w^{H_i}_{uv} \leq w_{uv}^G + x-1 \).

    \item[Construction of \( K_i \):]
    Let \( K_i \) be the graph resulting from augmenting \( H_i \) with an edge of weight \( x \cdot B \cdot h^2 \) between every pair \( \{u, v\} \) of nodes. If an edge $\{u,v\}$ exists, retain only the lighter of the two possible edge weights:
    \[
    w^{K_i}_{uv} = \min\left(w^{H_i}_{uv},\, x \cdot B \cdot h^2\right).
    \]

    \item[Construction of \( G_i \):]
    Let \( G_i \) be the graph obtained from \( K_i \) by dividing all edge weights by \( x \). Observe that all weights in \( G_i \) are integers bounded by \( B \cdot h^2 \), and the weighted diameter of $G_i$ is at most \( B \cdot h^2 \).
\end{description}

For the special case of \( i = 0 \), we have \( x = 1 \), so \( H_0 = G \), as the rounding has no effect. Similarly, we have \( K_0 = G_0 \). However, we emphasize that \( H_0 = G \) is \emph{not} the same as \( K_0 = G_0 \), due to the addition of an edge between every pair of nodes in the construction of \( K_0 \) from \( H_0 \). The main purpose of this step is to reduce the weighted diameter.

\paragraph{Computing $\eta(u,v)$.}
Given an \( l \)-approximation of APSP \( \delta_{G_i} \) for each  \( G_i \), and an \( h \)-approximation of APSP \( \delta \) for the original graph \( G \), each pair of nodes \( \{u, v\} \) computes the value \( \eta(u, v) \) as follows:

	\begin{itemize}
    \item If \( \delta(u, v) \geq \frac{B}{2} \cdot h^2 \), then choose \( i\) to be the unique integer such that \[ 2^{i-1} \cdot B \cdot h^2 \leq \delta(u, v) < 2^i \cdot B \cdot h^2, \] and compute \( \eta(u, v) = 2^i \cdot \delta_{G_i}(u, v) \). Since \( d(u, v) \in n^{O(1)} \), we have \( i \in O(\log n) \).
    
    \item If \( \delta(u, v) < \frac{B}{2} \cdot h^2 \), then choose \( i = 0 \) and compute \( \eta(u, v) =  2^i \cdot \delta_{G_i}(u, v) = \delta_{G_0}(u, v) \).
\end{itemize}

All computations are performed locally by the endpoints of each edge and therefore require zero communication rounds.

	\paragraph{Correctness.} To prove the correctness of the algorithm, consider any pair of nodes $\{u, v\}$ and recall that $x = 2^i$. We first show that $ \eta(u,v) \geq d_G(u,v)$.
\begin{align*}
    \eta(u,v) &= 2^i \cdot \delta_{G_i}(u, v) & \text{Definition of $\eta$}\\
    &\geq 2^i \cdot d_{G_i}(u, v) & \delta_{G_i}(u, v) \geq d_{G_i}(u, v)\\
    &\geq d_{K_i}(u, v) & \text{Construction of $G_i$}\\
    &= \min\left(d_{H_i}(u, v), 2^i \cdot B \cdot h^2\right) &\text{Construction of $K_i$}\\
 &\geq \min\left(d_G(u,v), 2^i \cdot B \cdot h^2\right) &\text{Construction of $H_i$}\\
    &= d_G(u,v). & d_G(u,v) \leq \delta(u,v) < 2^i \cdot B \cdot h^2
\end{align*}

Now, suppose there exists a shortest path between $u$ and $v$ with at most $h$ hops. Recall that in the construction of the graph $H_i$, each edge weight is additively increased by at most $x-1$, so \[d_{H_i}(u, v) \leq  d_G(u, v) + (x-1) \cdot h.\]
Our goal is to show that $d_{H_i}(u, v) \leq  (1+\varepsilon ) \cdot d_G(u, v)$. To do so, it suffices to upper bound the additive error $(x-1) \cdot h$ by $\varepsilon  \cdot d_G(u, v)$.  If \( \delta(u, v) < \frac{B}{2} \cdot h^2 \), then $x = 2^i = 1$, as $i = 0$, so the additive error $(x-1) \cdot h$ is zero, as required. 

For the rest of the proof, suppose  \( \delta(u, v) \geq \frac{B}{2} \cdot h^2 \). Observe that
    \[\frac{x}{2} \cdot B \cdot h^2 = 2^{i-1} \cdot B \cdot h^2 \leq \delta(u,v) \leq d_G(u,v) \cdot h,\]
    which implies \[d_G(u,v) \geq \frac{x}{2} \cdot B\cdot h = \frac{x}{2} \cdot \lceil2/\varepsilon \rceil \cdot h  \geq xh/\varepsilon.\]
    Therefore, we can upper bound the additive error $(x-1) \cdot h$ by
    \[ (x-1) \cdot h < xh \leq \varepsilon  \cdot d_G(u, v),\]
    as required.
\end{proof}

\paragraph{Remark.}
We briefly explain the intuition behind the \( h^2 \) factor in the weighted diameter bound \( \left\lceil \frac{2}{\varepsilon} \right\rceil \cdot h^2 \) for the graph \( G_i \), as follows. The initial \( h \)-approximation \( \delta \) is used to \emph{select} which \( \delta_{G_i} \) values should be considered. We are only interested in node pairs \( \{u, v\} \) for which there exists a shortest path of at most \( h \) hops. For such a path, each edge contributes a rounding error of at most \( 2^i \), so the total additive error is bounded by \( 2^i \cdot h \). We want this to be at most an \( \varepsilon \)-fraction of the actual distance \( d_G(u, v) \). However, since the approximation \( \delta \) may overestimate distances by a factor of up to \( h \), we must ensure that \( \delta(u, v) \geq \frac{2^i \cdot h^2}{\varepsilon} \) in order for the relative error to remain small.
In summary, one factor of \( h \) arises from the number of hops along the path, while the other factor of \( h \) stems from the approximation guarantee of \( \delta \).

\subsection{APSP Approximation With Large Bandwidth} 
\label{proof-special-case-2}

We now combine \Cref{weaker-approximate-APSP-first} and \Cref{hop-to-diameter} to show the following result.

\specialcaseii*

At a high level, the algorithm proceeds as follows:

\begin{enumerate}
    \item We begin by applying \Cref{constant-round-APSP-log-n} to obtain an $a \in O(\log n)$-approximation of APSP. Based on this, we compute a $\sqrt{n}$-nearest hopset $H$ with hop bound $\beta \in O(a \log n) \subseteq O(\log^2 n)$ using \Cref{approximate-to-hop}.
    
    \item Our next goal is to estimate distances to the $\sqrt{n}$-nearest nodes. The hopset $H$ ensures that each of these nodes is reachable via paths of at most $\beta \in O(\log^2 n)$ hops.
    \begin{enumerate}
        \item We construct $O(\log n)$ graphs $G_i$ by applying \Cref{hop-to-diameter} to $G \cup H$ with $h = \beta$. Each $G_i$ is guaranteed to have weighted diameter $O(h^2) \subseteq O(\log^4 n)$.
        
        \item We then apply \Cref{weaker-approximate-APSP-first} to each $G_i$, obtaining an $O(1)$-approximation of all-pairs distances in $G_i$ within $O(\log \log \log n)$ rounds. By \Cref{hop-to-diameter}, this gives an $O(1)$-approximation of the distances in the original graph $G$ for node pairs connected by $h$-hop paths in $G \cup H$. In particular, each node $u$ computes a set $\tilde{N}_k(u)$ of approximately its $k = \sqrt{n}$-nearest nodes, satisfying the conditions required in \Cref{skeleton}.
    \end{enumerate}

    \item Finally, we apply \Cref{skeleton} to construct a skeleton graph $G_S$ on $\tilde{O}(\sqrt{n})$ nodes, such that an $O(1)$-approximation of distances in $G_S$ implies an $O(1)$-approximation in $G$. Since $G_S$ is small, we can compute this approximation in $O(1)$ rounds. Optimizing parameters yields a final distance approximation of $7^3 + \varepsilon$ in $G$.
\end{enumerate}

The bandwidth \( B = \log^4 n \) in \Cref{thm:special-case-2-approximate-APSP} makes sense because, to implement this approach, we need to run the algorithm from \Cref{weaker-approximate-APSP-first}---which operates in the \( \clique[\log^3 n] \) model---on \( O(\log n) \) different graphs \( G_i \).

\begin{proof}[Proof of \Cref{thm:special-case-2-approximate-APSP}] 
In the proof, we only compute a $7^3(1+\varepsilon)^2$-approximation. By a change of variable $\varepsilon' = 7^3(2\varepsilon + \varepsilon^2)$, the approximation factor becomes $7^3 + \varepsilon'$, as required.
We begin with the algorithm description.

\paragraph{Step 1: Initialization.} We begin by applying \Cref{constant-round-APSP-log-n} to compute an $a$-approximation of APSP, where $a \in O(\log n)$.

Next, we use \Cref{approximate-to-hop} with approximation factor $a \in O(\log n)$ and weighted diameter bound $d \in n^{O(1)}$ to construct a $\sqrt{n}$-nearest $\beta$-hopset $H$, where $\beta \in O(a \log d) \subseteq O(\log^2 n)$.

By the definition of hopsets, any node can reach its $\sqrt{n}$-nearest nodes using at most $\beta$ hops in $G \cup H$, and for any pair of nodes $\{u, v\}$, the exact distances are preserved: $d_G(u, v) = d_{G \cup H}(u, v)$. Hence, we may work with the augmented graph $G \cup H$ for the remainder of the algorithm.

\paragraph{Step 2(a): Weight scaling lemma.} Our next goal is to estimate distances to the $\sqrt{n}$-nearest nodes. To achieve this, it suffices to approximate $d_{G \cup H}(u, v)$ for all node pairs $\{u, v\}$ that are reachable within $\beta$ hops in $G \cup H$. This is done using \Cref{hop-to-diameter} with the following parameters:
\begin{itemize}
    \item Set $h = \beta \in O(\log^2 n)$.
    \item Set $l = 7$, as we later apply \Cref{weaker-approximate-APSP-2} to compute a 7-approximation of APSP for each $G_i$.
    \item The input graph is $G \cup H$.
    \item The distance estimate $\delta$ is the $O(\log n)$-approximation of $G \cup H$ obtained from \Cref{constant-round-APSP-log-n}.
    \item Let $\varepsilon > 0$ be an arbitrarily small constant.
\end{itemize}

As a result, we obtain a sequence of $O(\log n)$ graphs $\{ G_0, G_1, \ldots, G_{O(\log n)} \}$. By the guarantees of \Cref{hop-to-diameter}, the weighted diameter of each graph $G_i$ is $O(h^2) \subseteq O(\log^4 n)$.

\paragraph{Step 2(b): Distance estimation.} We now apply the algorithm from \Cref{weaker-approximate-APSP-2} to each of the graphs $G_i$, computing a 7-approximation $\delta_{G_i}$ of distances in each graph.

This is doable because the increased bandwidth in the $\clique[\log^4 n]$ allows us to run $O(\log n)$ instances of the algorithm from \Cref{weaker-approximate-APSP-2}, which operates in the $\clique[\log^3 n]$ model, in parallel. This step takes $O(\log \log \log n)$ rounds.

By the conclusion of \Cref{hop-to-diameter}, combining the 7-approximation $\delta_{G_i}$ of each $G_i$ yields a global distance estimate $\eta$ on $G \cup H$. 
For every pair of nodes $\{u, v\}$, the estimate $\eta$ satisfies:
\begin{equation}
d_G(u, v) = d_{G \cup H}(u, v) \leq \eta(u, v). \label{myeq1}
\end{equation}
Moreover, if there exists a shortest path between $u$ and $v$ in $G \cup H$ with at most $\beta$ hops, then the estimate $\eta$ additionally satisfies:
\begin{equation}
\eta(u, v) \leq 7(1+\varepsilon) \cdot d_{G \cup H}(u, v) = 7(1+\varepsilon) \cdot d_G(u, v). \label{myeq}
\end{equation}

In particular, by the construction of $H$, for every node $u$ and every node $v \in N_{\sqrt{n}}(u)$, there exists a path of at most $\beta$ hops in $G \cup H$. Therefore, \Cref{myeq} applies to all such pairs $u$ and $v$.

\paragraph{Step 3: Skeleton graph construction.} Finally, we apply \Cref{skeleton} with the following parameters:
\begin{itemize}
    \item Set $\delta$ to be the $7(1+\varepsilon)$-approximation $\eta$ computed in the previous step, so $a = 7(1+\varepsilon)$.
    \item Set $k = \sqrt{n}$.
    \item For each node $u$, define $\tilde{N}_k(u)$ as the set of $k$ nodes $v$ with the smallest values of $\eta(u, v)$, breaking ties by node $\ID$s.
    \item Set $l = 1$, since we will compute the \emph{exact} APSP of $G_S$ via brute force.
\end{itemize}

The resulting skeleton graph $G_S$ has $O(|V_S|^2) \subseteq O\left(\frac{n^2 \log^2 k}{k^2}\right) = O(n \log^2 n)$ edges. As such, all edges of $G_S$ can be broadcast in $O(1)$ rounds, after which each node can compute the exact APSP of $G_S$ locally. Therefore, by \Cref{skeleton}, we obtain a $7 l a^2 = 7^3 \cdot (1+\varepsilon)^2$-approximation of APSP in $G$, as desired.

\paragraph{Correctness.} For the remainder of the proof, we omit the subscript and simply write \( d(u,v) \) instead of \( d_G(u,v) = d_{G \cup H}(u,v) \), since distances in \( G \) and \( G \cup H \) are identical by construction.

It remains to verify that the distance estimate \( \delta = \eta \), along with the chosen sets \( \tilde{N}_k(u) \), satisfies the following two conditions required by \Cref{skeleton}:

\begin{description}
    \item[(C1)] For each node \( v \in \tilde{N}_k(u) \), the value \( \delta(u, v) \) is known to \( u \), and satisfies
    \[
    d(u,v) \leq \delta(u, v) \leq a \cdot d(u, v).
    \]
    \item[(C2)] For any \( v \in \tilde{N}_k(u) \) and \( t \notin \tilde{N}_k(u) \), we have
    \[
    \delta(u, v) \leq a \cdot d(u, t).
    \]
\end{description}
Recall the sets \( \tilde{N}_k(u) \) are computed with respect to the approximate distances \( \delta = \eta \), and thus may differ from the true set \( N_k(u) \) of the exact \( k \)-nearest nodes. Nevertheless, we will show that the conditions above still hold.

To prove that (C1) and (C2) hold, we first define a node \( x \) as follows. If \( N_k(u) \setminus \tilde{N}_k(u) \neq \emptyset \), then let \( x \) be a node in this set with the smallest value of \( d(u, x) \). Otherwise, since both \( N_k(u) \) and \( \tilde{N}_k(u) \) contain exactly \( k \) nodes, we must have \( N_k(u) = \tilde{N}_k(u) \). In this case, let \( x \) be a node in \( N_k(u) \) with the largest value of \( \delta(u, x) \).

In either case, we have the key property that \( \delta(u, v) \leq \delta(u, x) \) for every \( v \in \tilde{N}_k(u) \). Moreover, since \( x \in N_k(u) \), we have \( \delta(u, x) \leq a \cdot d(u, x) \) by \Cref{myeq}.

\paragraph{Proof of (C1).} We divide the analysis into two cases.
\begin{itemize}
    \item If \( v \in N_k(u) \), then \(d(u, v) \leq \delta(u, v) \leq a \cdot d(u, v) \) by \Cref{myeq1,myeq}.
    \item Otherwise, $d(u, v) \leq \delta (u, v) \leq  \delta (u, x) \leq  a \cdot d(u, x) \leq  a \cdot d(u, v)$, where the first inequality follows from \Cref{myeq1}, and the last inequality holds because $x \in N_k(u)$ and $v \notin N_k(u)$, so \( d(u, x) \leq d(u, v) \).
\end{itemize}

\paragraph{Proof of (C2).}  Let \( v \in \tilde{N}_k(u) \) and \( t \notin \tilde{N}_k(u) \). Then,
\[
\delta(u, v) \leq \delta(u, x) \leq a \cdot d(u, x) \leq a \cdot d(u, t).
\]
Similarly, the final inequality holds because \( x \in N_k(u) \) and \( t \notin N_k(u) \), so \( d(u, x) \leq d(u, t) \).
 \end{proof}

 \paragraph{Remark.}
 In the proof above, the only step that requires $O(\log \log \log n)$ rounds is \Cref{weaker-approximate-APSP-2}; all remaining steps can be performed in $O(1)$ rounds. There are two places where extra bandwidth is necessary: first, in the parallel execution of \Cref{weaker-approximate-APSP-2} across $O(\log n)$ graphs; and second, in the final step where all edges of the skeleton graph are broadcast.

\subsection{APSP Approximation With Small Bandwidth} 
\label{subsec:general_case}

Here, we make the algorithm
described in \Cref{proof-special-case-2} work in the standard \clique model
while only adding a constant factor to the approximation.

\main*

At a high level, we begin by computing the distances to the $\log^4 n$-nearest nodes for each node. This can be done in $O(1)$ rounds using \Cref{lemma:fast-matrix-exp-i-iter}. We then construct a skeleton graph $G_S$ with $O(n / \log^3 n)$ nodes and simulate the algorithm of \Cref{thm:special-case-2-approximate-APSP} on $G_S$.
Since $G_S$ has $O(n / \log^3 n)$ nodes, we can simulate routing $O(\log^4{n})$ bits between each of its nodes in $O(1)$ rounds using \Cref{lemma:Lenzen-routing}.

\begin{proofof}{\Cref{approximate-APSP}}
We begin by computing, for each node \( v \), the distances between $v$ and the set \( N_k(v) \) of its \( k = \log^4 n \)-nearest nodes, in \( O(1) \) rounds. This is achieved by invoking \Cref{lemma:fast-matrix-exp-i-iter} with the following parameters:
\begin{itemize}
    \item Set \( k = \log^4 n \).
    \item Choose \( h \in \Theta\left(\frac{\log n}{\log \log n}\right) \) so that \( k \in O(n^{1/h}) \).
    \item Set \( i \in O(1) \) such that \( k \le h^i \).
\end{itemize}
Since each node can reach its \( k \)-nearest nodes within at most \( k \) hops, the output \( N_k^{h^i}(v) \) produced by the algorithm matches exactly the set \( N_k(v) \) of $k$-nearest nodes.

Next, we apply \Cref{lemma_skeleton_simplified} with parameters \( k = \log^4 n \) and \( l = 7^3 \cdot (1+\varepsilon)^2 \). This produces a skeleton graph \( G_S \) consisting of \( O\left(\frac{n \log k}{k}\right) \subseteq O\left(\frac{n}{\log^3 n}\right) \) nodes.

We can simulate an algorithm on \( G_S \) in the \( \clique[\log^4 n] \) model using \Cref{lemma:Lenzen-routing}, incurring only an \( O(1) \)-factor overhead. Specifically, each node in \( G_S \) needs to send and receive \( O(\log^4 n) \) bits to and from each of the \( O\left(\frac{n}{\log^3 n}\right) \) other nodes. This results in a total communication volume of \( O(n) \) messages of \( O(\log n) \) bits per node, which can be routed in \( O(1) \) rounds via \Cref{lemma:Lenzen-routing}.

We apply the algorithm of \Cref{thm:special-case-2-approximate-APSP} to compute an \( l = (7^3 + \varepsilon) \)-approximate APSP on \( G_S \). By the guarantee of \Cref{lemma_skeleton_simplified}, this implies that each node can compute a \( 7l = 7(7^3 + \varepsilon) \)-approximate APSP on the original graph \( G \). By setting \( \varepsilon' = 7\varepsilon \), the overall approximation factor becomes \( 7^4 + \varepsilon' \), as claimed.

All steps aside from the invocation of \Cref{thm:special-case-2-approximate-APSP} complete in \( O(1) \) rounds. The algorithm of \Cref{thm:special-case-2-approximate-APSP} runs in \( O(\log \log \log n) \) rounds, and therefore the overall round complexity is \( O(\log \log \log n) \), as required.
\end{proofof}

\subsection{Limiting the Number of Rounds}
\label{subsec:truncated-algorithm}

We now demonstrate a tradeoff between the round complexity and the approximation factor by restricting the algorithm to run in \( O(t) \) rounds.

\truncatedalgorithm*

Recall that \Cref{approximate-APSP} builds on \Cref{thm:special-case-2-approximate-APSP}, which itself is based on \Cref{weaker-approximate-APSP-first}. We therefore begin by considering round-limited versions of these results.

\begin{lemma}[{Limiting the number of rounds of \Cref{weaker-approximate-APSP-first}}]
	\label{weaker-approximate-APSP-truncated}
Given that $G$ is a weighted undirected graph with  weighted diameter $d  \in  (\log n)^{O(1)}$, for any $t\geq 1$, there is an $O\left(\log^{2^{-t}} n\right)$-approximation algorithm for APSP that takes $O(t)$ rounds.
\end{lemma}
\begin{proof}
We repeat the proof of \Cref{weaker-approximate-APSP-first}, with the following modifications:
\begin{itemize}
	\item If $t$ is large enough such that the desired approximation factor $O\left(\log^{2^{-t}} n\right)$ is $  O(\log \log n)$, then $t  \in  \Omega (\log \log \log n)$, so we can apply \Cref{weaker-approximate-APSP-first} without modifications, as it already finishes in $O(\log \log \log n) \subseteq  O(t)$ rounds.
	\item Otherwise, we only use the first part of the algorithm of \Cref{weaker-approximate-APSP-first} and limit the number of applications of \Cref{approximate-reduction} to $t$, so we obtain an $O(\log^{2^{-t}} n)$-approximation of APSP in $O(t)$ rounds.\qedhere
\end{itemize}
\end{proof}

\begin{lemma}[{Limiting the number of rounds of \Cref{thm:special-case-2-approximate-APSP}}]
	\label{thm:special-case-2-approximate-APSP-truncated}
For any $t \geq 1$, an $O\left(\log^{2^{-t}} n\right)$-approximation of APSP in weighted undirected graphs can be computed w.h.p.\ in $O(t)$ rounds in the $\clique[\log^4 n]$ model.
\end{lemma}
\begin{proof}
We proceed in the same manner as the proof of \Cref{thm:special-case-2-approximate-APSP}, with a modification. Instead of applying \Cref{weaker-approximate-APSP-first} to obtain a 7-approximation \( \delta_{G_i} \) on each graph \( G_i \) in \( O(\log \log \log n) \) rounds, we apply \Cref{weaker-approximate-APSP-truncated} with parameter \( t' = t + 1 \) to obtain a \( b \in O\left(\log^{2^{-t'}} n\right) = O\left(\sqrt{\log^{2^{-t}} n}\right) \)-approximation \( \delta_{G_i} \) in \( O(t') = O(t) \) rounds.

As a result, in the skeleton graph construction, the parameter \( a \) becomes \( a = b(1+\varepsilon) \in O\left(\sqrt{\log^{2^{-t}} n}\right) \), while the other parameters remain unchanged. The final approximation ratio is then \( 7la^2 \in O(b^2) \subseteq O\left(\log^{2^{-t}} n\right) \), as desired.
\end{proof}

\begin{proof}[Proof of \Cref{tradeoff-APSP}]
The proof of \Cref{tradeoff-APSP} closely follows that of \Cref{approximate-APSP}, with the only difference being that we apply \Cref{thm:special-case-2-approximate-APSP-truncated} with parameter $t$ in place of \Cref{thm:special-case-2-approximate-APSP}. As a result, the overall round complexity improves to $O(t)$, while the final approximation ratio becomes $7 \cdot O\left(\log^{2^{-t}} n\right) = O\left(\log^{2^{-t}} n\right)$.
\end{proof}

\section{Discussion}

In this work, we show that an \( O(1) \)-approximation for APSP can be computed in \( O(\log \log \log n) \) rounds in the \clique model. Many intriguing questions remain open.

A natural next goal is to achieve an \( O(1) \)-approximation in \( O(1) \) rounds. One promising direction is to design faster algorithms for computing distances to the \( k \)-nearest nodes. Our current approach computes these distances in \( O(1) \) rounds \emph{only} for sub-polynomial values of \( k \). If one can compute these distances, or even an \( O(1) \)-approximation of them, in \( O(1) \) rounds for a small polynomial \( k \), this would lead to an \( O(1) \)-approximation for APSP in \( O(1) \) rounds.

Another important direction is to improve the approximation factor. Although our algorithm achieves a constant-factor approximation, the constant \( 7^4 + \varepsilon = 2401+ \varepsilon\) is very large and unlikely to be significantly reduced through simple optimizations. An intriguing open question is whether one can design a generic  procedure to improve the approximation factor: Given any \( t \)-round \( O(1) \)-approximation algorithm as a black box, automatically improve it to an \( O(t) \)-round \( c \)-approximation algorithm, for some small constant \( c \), such as 2 or 3.

It is also interesting to ask whether similar results can be obtained in the closely related linear-memory Massively Parallel Computation (MPC) model. A direct simulation of the \clique model in near-linear-memory MPC requires \( \tilde{\Omega}(n^2) \) total memory, which exceeds the typical memory bound of \( \tilde{O}(m) \) in the MPC model. A promising direction, therefore, is to develop low-memory variants of our algorithm that are suitable for the MPC model.

Finally, the new techniques introduced in this work, such as our new skeleton graphs, may be of independent interest. It would be worthwhile to explore their applications in other models of computation, such as the \congest model of distributed computing.

\bibliographystyle{alpha}
\bibliography{References}

\appendix

\section{Proof of \texorpdfstring{\Cref{thm:zero_weight_component_compression}}{Theorem 2.1}} \label{sec:app_prelim}

In the appendix, we prove \Cref{thm:zero_weight_component_compression}.

\thmZero*

\begin{proof}
The algorithm is as follows:
\begin{enumerate}
\item Identify the connected components of the subgraph formed by the edges of weight $0$. That is, nodes $u$ and $v$ belong to the same component if and only if $d_G(u, v) = 0$.
\item For each connected component, select a representative node to serve as its leader.
\item For every pair of leaders $s$ and $t$, have both $s$ and $t$ learn the minimum weight of an edge connecting their respective components, if such an edge exists.
\end{enumerate}
All of the above can be completed in $O(1)$ rounds, as follows:
\begin{description}
    \item[Step 1:] Computing the connected components formed by the zero-weight edges can be done by first computing a minimum spanning tree using the algorithm of Nowicki~\cite{10.1145/3406325.3451136}, which runs in $O(1)$ rounds.  By the end of the algorithm, every node learns the entire minimum spanning tree. With this global view, each node can locally filter out all nonzero-weight edges and retain only the zero-weight edges in the spanning tree, allowing them to identify the connected components induced by the zero-weight edges.

    \item[Step 2:] For each connected component, designate the node with the smallest $\ID$ as its leader.
Since all nodes know the full structure of the connected components, this step can also be performed locally without communication.

    \item[Step 3:] This step is nontrivial and is discussed below.
\end{description}

Let $S$ be the set of leaders.  
For each leader $s \in S$, let $C(s) = \{v \in V \mid d(s, v) = 0\}$ denote the set of nodes in the same connected component as $s$, including $s$ itself. 
By the discussion above, every node knows the value of $C(s)$ for all $s \in S$.

	We execute the following algorithm:
	\begin{enumerate}
		\item For each node $v$ and each leader $t$, node $v$ sends a message $(s, w)$ to node $t$, where $s$ is the leader of the connected component of $v$, and $w$ is the minimum weight of an edge connecting $v$ and any node in $C(t)$.
		\item Each leader $t$ receives a set of messages $(s, w)$. For each value of $s$, the minimum value of $w$ among all these messages is the minimum-weight edge connecting any node in $C(t)$ to any node in $C(s)$.
	\end{enumerate}
    As each node $v$ just sends one message to each leader $t$, the above algorithm can be completed in one round.

Afterwards, we consider the compressed graph whose node set is the set $S$ of leaders, and the weight of the edge between two leaders $s$ and $t$ is the minimum weight of an edge connecting  $C(t)$ and $C(s)$.

We run the algorithm $\mathcal{A}$ on this graph to obtain a distance approximation  
\[
\delta \colon S \times S \to \mathbb{Z}
\]  
in $f(n)$ rounds. We then compute the final distance estimate  
\[
\eta \colon V \times V \to \mathbb{Z}
\]  
as follows: For any two leaders $s, t \in S$, and any nodes $v \in C(s)$ and $u \in C(t)$, we set $\eta(v, u) = \delta(s, t)$.  
Clearly, $\eta$ is an $a$-approximation of APSP on $G$.

Every node already knows all sets $C(s)$, so to complete the algorithm, each node $v \in C(s)$ only needs to learn the values $\delta(s, t)$ for all $t \in S$.  
This is done by having each leader $t$ send the value $\delta(s, t)$ to every node $v \in C(s)$.  
Since each leader sends only one message to each node, this final step takes $O(1)$ rounds.
\end{proof}

\end{document}